\documentclass[a4paper, USenglish, cleveref, autoref, thm-restate, nosubfigcap]{lipics-v2021}
\usepackage{amsmath, amssymb, amsfonts}

\usepackage{tcolorbox}
\tcbuselibrary{skins}
\usepackage{xcolor}

\nolinenumbers

\colorlet{mix}{red!50!black}

\usepackage{hyperref}
\hypersetup{colorlinks={true},linkcolor={mix},citecolor=mix}

\usepackage{algorithm2e}
\RestyleAlgo{ruled}
\usepackage{booktabs}
\usepackage{xspace}
\usepackage{mathtools}
\usepackage{thmtools} 
\usepackage{thm-restate}

\newtheorem{reduction rule}{Reduction Rule}
\newtheorem*{reduction rule*}{Reduction Rule}

\def\zzcommand#1{\let#1\undefined\newcommand#1}

\newtheorem{reduction*}[theorem]{Reduction Rule}

\Crefname{reduction*}{Reduction Rule}{Reduction Rules}
\Crefname{step}{Step}{Steps}
\Crefname{enumi}{Step}{Steps}

\crefname{theorem}{Thm.}{Thms.}
\Crefname{theorem}{Theorem}{Theorems}
\crefname{corollary}{Cor.}{Cors.}
\Crefname{corollary}{Corollary}{Corollaries}

\DeclareDocumentCommand{\set}{m g o}%
{%
    \IfNoValueTF{#3}{\left}{#3}\{#1
           \IfNoValueTF{#2}{}{\ \IfNoValueTF{#3}{\left}{#3}\vert\ \vphantom{#1}#2\IfNoValueTF{#3}{\right.}{}}
                \IfNoValueTF{#3}{\right}{#3}\}%
}

\DeclareDocumentCommand{\abs}{m o}%
{%
    \IfNoValueTF{#2}{\left}{#2}\vert#1
                \IfNoValueTF{#2}{\right}{#2}\vert%
}

\newcommand{\N}{\mathbb{N}}
\newcommand{\bigO}[1]{\mathcal{O}(#1)}

\newcommand{\shor}[1]{P_{U}({#1})}

\DeclareMathOperator{\vOp}{V}
\DeclareMathOperator{\eOp}{E}
\newcommand*{\ve}[1]{\vOp({#1})}
\newcommand*{\e}[1]{\eOp({#1})}

\newcommand{\uned}[2]{\set{{#1}, {#2}}}
\newcommand{\died}[2]{({#1}, {#2})}

\DeclareMathOperator{\neiOp}{N}
\DeclareMathOperator{\outNeiOp}{\neiOp^+}
\newcommand{\nei}[1]{\neiOp({#1})}
\newcommand{\cnei}[1]{\neiOp[{#1}]}
\newcommand{\outNei}[1]{\outNeiOp({#1})}
\newcommand{\coutNei}[1]{\outNeiOp[{#1}]}
\newcommand{\inNei}[1]{\neiOp^-({#1})}

\newcommand{\neiG}[2]{\neiOp_{#1}({#2})}

\newcommand{\outNeiG}[2]{\neiOp^+_{#1}({#2})}

\newcommand{\inNeiG}[2]{\neiOp^{-}_{#1}({#2})}

\newcommand*{\induced}[2]{{#1}[{#2}]}

\DeclareMathOperator{\exOp}{\textrm{Ext}}
\newcommand*{\ex}[3]{\exOp_{#1}({#2}, {#3})}

\newcommand*{\exset}[2]{E_{{#1}{#2}}}
\newcommand*{\expair}[2]{({#2}, \exset{#1}{#2})}
\newcommand*{\dexset}{\exset{G}{D}}
\newcommand*{\dexpair}{\expair{G}{D}}

\newcommand*{\expairi}[1]{\expair{G}{D_{#1}}}

\DeclareMathOperator{\sccOp}{\textrm{SCC}}
\newcommand*{\scc}[1]{\sccOp({#1})}

\DeclareMathOperator{\wOp}{\omega}
\newcommand{\w}[1]{\wOp({#1})}

\DeclareMathOperator{\compOp}{\textrm{Comp}}
\newcommand{\compG}[3]{\compOp_{#1}({#2}, {#3})}

\newcommand*{\dcompG}{\compG{G}{D}{\dexset}}

\DeclareMathOperator{\sourceOp}{\mathcal{S}}
\newcommand{\source}[1]{\sourceOp({#1})}
\newcommand{\sourcee}[2]{\sourceOp({#1}, {#2})}
\DeclareMathOperator{\sinkOp}{\mathcal{T}}
\newcommand{\sink}[1]{\sinkOp({#1})}
\newcommand{\sinke}[2]{\sinkOp({#1}, {#2})}
\DeclareMathOperator{\connOp}{\mathcal{C}}
\newcommand{\conn}[1]{\connOp({#1})}
\newcommand{\conne}[2]{\connOp({#1}, {#2})}

\DeclareMathOperator{\propOp}{\Pi}

\newcommand{\prname}[2]{\textsc{{#1}-Secluded {#2}}}

\SetKwInput{KwData}{Input}
\SetKwInput{KwResult}{Output}

\usepackage{tabularx, environ}

\makeatletter

\newcolumntype{\expand}{}
\long\@namedef{NC@rewrite@\string\expand}{\expandafter\NC@find}

\NewEnviron{problem}[2][]{%
  \def\problem@arg{#1}%
  \def\problem@framed{framed}%
  \def\problem@hline{\hline}%
\def\problem@tablelayout{|>{\bfseries}lX|c}%
\def\problem@title{\multicolumn{2}{|%
  >{\raisebox{-\fboxsep}}%
  p{\dimexpr\textwidth-4\fboxsep-2\arrayrulewidth\relax}%
  |}{%
      {\bfseries Problem.} \textsc{#2}%
  }}%
  \bigskip\par\noindent%
  \begin{tabularx}{\textwidth}{\expand\problem@tablelayout}%
    \problem@hline%
    \problem@title\\[2\fboxsep]%
    \BODY
    \\\problem@hline%
  \end{tabularx}%
  \medskip\par%
  \vspace{.5mm}
}
\makeatother

\makeatletter
\providecommand{\@fourthoffour}[4]{#4}
\newcommand\fixstatement[2][\proofname\space of]{%
  \ifcsname thmt@original@#2\endcsname
    \AtEndEnvironment{#2}{%
      \xdef\pat@label{\expandafter\expandafter\expandafter
        \@fourthoffour\csname thmt@original@#2\endcsname\space\@currentlabel}%
      \xdef\pat@proofof{\@nameuse{pat@proofof@#2}}%
    }%
  \else
    \AtEndEnvironment{#2}{%
      \xdef\pat@label{\expandafter\expandafter\expandafter
        \@fourthoffour\csname #1\endcsname\space\@currentlabel}%
      \xdef\pat@proofof{\@nameuse{pat@proofof@#2}}%
    }%
  \fi
  \@namedef{pat@proofof@#2}{#1}%
}

\title{A Parameterized Study of Secluded Structures in Directed Graphs} 

\titlerunning{A Parameterized Study of Secluded Structures in Directed Graphs} 
\author{Jonas Schmidt}{Bocconi University, Milan, Italy}{jonas.schmidt2@phd.unibocconi.it}{https://orcid.org/0000-0002-1115-3868}{}
\author{Shaily Verma}{Hasso Plattner Institute, Potsdam, Germany}{shaily.verma@hpi.de}{https://orcid.org/0009-0000-6789-1643}{}
\author{Nadym Mallek}{Hasso Plattner Institute, Potsdam, Germany}{nadym.mallek@hpi.de}{https://orcid.org/0000-0002-4370-5145}{}

\authorrunning{J. Schmidt, S. Verma, and N. Mallek} 

\Copyright{Jonas Schmidt, Shaily Verma, and Nadym Mallek} 
\hideLIPIcs



\begin{document}

\maketitle

\begin{abstract}
Given an undirected graph $G$ and an integer $k$, the \textsc{Secluded $\Pi$-Subgraph} problem asks you to find a maximum size induced subgraph that satisfies a property $\Pi$ and has at most $k$ neighbors in the rest of the graph. This problem has been extensively studied; however, there is no prior study of the problem in directed graphs. This question has been mentioned by Jansen et al. [ISAAC'23].

In this paper, we initiate the study of \textsc{Secluded Subgraph} problems in directed graphs by incorporating different notions of neighborhoods: in-neighborhood, out-neighborhood, and their union. Formally, we call these problems \textsc{\{In, Out, Total\}-Secluded $\Pi$-Subgraph}, where given a directed graph $G$ and an integer $k$, we want to find an induced  subgraph satisfying $\Pi$ of maximum size that has at most $k$ in/out/total-neighbors in the rest of the graph, respectively.
We investigate the parameterized complexity of these problems for different properties $\Pi$.
In particular, we prove the following parameterized results:
\begin{itemize}
    \item We design an FPT algorithm for the \textsc{Total-Secluded Strongly Connected Subgraph} problem when parameterized by $k$. 
    \item We show that the \textsc{Out-Secluded $\mathcal{F}$-Free Subgraph} problem with parameter $k$ is W[1]-hard, where $\mathcal{F}$ is a family of directed graphs except any subgraph of a star graph whose edges are directed towards the center. This result also implies that \textsc{In/Out-Secluded DAG} is W[1]-hard, unlike the undirected variants of the two problems, which are FPT.
    \item We design an FPT-algorithm for \textsc{In/Out/Total-Secluded $\alpha$-Bounded Subgraph} when parameterized by $k$, where $\alpha$-bounded graphs are a superclass of tournaments. 
    \item For undirected graphs, we improve the best-known FPT algorithm for \textsc{Secluded Clique} by providing a faster FPT algorithm that runs in time $1.6181^kn^{\bigO{1}}$. 
\end{itemize}



\end{abstract}

\newpage

\section{Introduction}
Finding substructures in graphs that satisfy specific properties is a ubiquitous problem. This general class of problems covers many classical graph problems such as finding maximum cliques, Steiner trees, or even shortest paths. Another compelling property to look for in a substructure is its isolation from the remaining graph. This motivated Chechik et al.~\cite{chechik2017secluded}, to introduce the concept of \emph{secluded subgraphs}. 
Formally, in the \textsc{Secluded $\propOp$-Subgraph} problem, given an undirected graph $G$, the goal is to find a maximum size subset of vertices $S \subseteq \ve{G}$ such that the subgraph induced on $S$ fulfills a property $\propOp$ and has a neighborhood $\abs{\nei{S}} \le k$ where $k$ is a natural number.

This problem has been studied extensively for various properties $\propOp$ such as paths~\cite{van2020parameterized,luckow2020computational,fomin2017parameterized,chechik2017secluded}, Steiner trees~\cite{chechik2017secluded,fomin2017parameterized}, induced trees~\cite{donkers2023finding,golovach2020finding}, subgraphs free of forbidden induced subgraphs~\cite{golovach2020finding,jansen2023single}, and more~\cite{bevern207finding}.
Most of these studies focus on the parameterized setting, due to the strong relation to vertex deletion and separator problems, which are foundational in parameterized complexity. Our problem fits in that category since the neighborhood can be considered a $(S,V \setminus S)$-separator.

While the undirected \textsc{Secluded $\Pi$-Subgraph} problem has been explored and widely understood in prior work~\cite{jansen2023single,golovach2020finding,donkers2023finding}, the directed variant has not yet been studied, although it is a natural generalization and was mentioned as an interesting direction by Jansen et al.~\cite{jansen2023single}. Directed graphs naturally model real-world systems with asymmetric interactions, such as social networks with unidirectional follow mechanisms or information flow in communication systems. Furthermore, problems such as \textsc{Directed Feedback Vertex Set}, \textsc{Directed Multicut}, and \textsc{Directed Multiway Cut} underline how directedness can make a fascinating and insightful difference when it comes to parameterized complexity ~\cite{chen2008fixed,hatzel2023fixed,multicut4,chitnis2013fixed}. These problems and their results provide a ground for studying directed secluded subgraph problems, motivating the need to investigate them systematically.

In this paper, we introduce three natural directed variants of the \textsc{Secluded $\Pi$-Subgraph} problem, namely \textsc{Out-Secluded $\Pi$-Subgraph}, \textsc{In-Secluded $\Pi$-Subgraph}, and \textsc{Total-Secluded $\Pi$-Subgraph}. These problems limit either the out-neighborhood of $S$, its in-neighborhood, or the union of both. The out/in-neighborhood of a set $S$ is the set of vertices in $\ve{G} \setminus S$ reachable from $S$ via an outgoing/incoming edge. 
The problems corresponding to different types of neighborhoods can be encountered in real-life networks. In privacy-aware social network analysis, one might aim to identify a community with minimal external exposure. Similarly, robust substructures with limited connectivity to vulnerable components are critical in cybersecurity. 
These real-world motivations further emphasize the need to explore and formalize directed variants of the \textsc{Secluded $\Pi$-Subgraph} problem. 

\begin{table}[t]
    \centering
    \begin{tabular}{lll}
        \toprule
        \textbf{Property $\propOp$} & \textbf{In- / Out-Secluded} & \textbf{Total-Secluded}\\
        \midrule
        \textsc{WC $\mathcal{F}$-Free Subgraph} & W[1]-hard \hfill \cref{thm:f-free-hard-always} & FPT$^*$ \hfill\cite{jansen2023single}\\
        \textsc{WC DAG} & W[1]-hard \hfill \cref{cor:dag_out_k+t} & ?\\
        \textsc{$\alpha$-Bounded Subgraph} & FPT \hfill \cref{thm:alpha_bounded_fpt}& FPT \hfill \cref{thm:alpha_bounded_total}\\
        \textsc{Strongly Connected Subgraph} & ? & FPT \hfill \cref{cor:scc_algo}\\
        \midrule
        \textsc{Clique} & \multicolumn{2}{l}{FPT in time $1.6181^kn^{\bigO{1}}$ \hfill \cref{thm:clique_better}}\\
        \bottomrule
    \end{tabular}
    \caption{Our main results for directed and undirected problems. WC stands for weakly connected. All FPT results are with respect to parameter $k$ and all hardness results are with respect to parameter $k+w$. For the entry marked with $*$, the undirected algorithm from~\cite{jansen2023single} immediately generalizes to the total-secluded setting (for finite $\mathcal{F}$). The problems marked with $(?)$ are open. 
    \label{tab:our_results}}
\end{table}

For any property $\propOp$, we formulate our general problem as follows. Notice that in-secluded and out-secluded are equivalent for all properties that are invariant under the transposition of the edges. For this reason, we mostly focus on total-neighborhood and out-neighborhood.
\begin{tcolorbox}[enhanced,title={\color{black} {\textsc{X-Secluded $\propOp$-Subgraph} \quad ($\text{X} \in \set{\text{In}, \text{Out}, \text{Total}}$)}}, colback=white, boxrule=0.4pt,
	attach boxed title to top left={xshift=.3cm, yshift*=-2.5mm},
	boxed title style={size=small,frame hidden,colback=white}]
	\textbf{Parameter:} An integer $k \in \N$ 

	\textbf{Input:}  A directed graph $G$ with vertex weights $\omega \colon V \to \N$ and an integer $w \in \N$

	\textbf{Output:} An $k$-X-secluded set $S \subseteq \ve{G}$ of weight $\omega(S) \ge w$ that satisfies $\propOp$, or report that none exists. 
\end{tcolorbox} 

\subsection*{Our Contribution}
In this subsection, we present the results we obtained for the problem. See \Cref{tab:our_results} for an overview of the results we obtained in this paper. 

\subparagraph*{Strongly Connected Subgraph.} We show that the \textsc{Total-Secluded Strongly Connected Subgraph} problem is fixed-parameter tractable when parameterized with $k$. Precisely, we prove the following result:

\begin{restatable}[]{theorem}{test}
\label{cor:scc_algo}
  \textsc{Total-Secluded Strongly Connected Subgraph} is solvable in time $2^{2^{2^{\bigO{k^2}}}}n^{\bigO{1}}$.
\end{restatable}

We design our FPT algorithm for \textsc{Total-Secluded Strongly Connected Subgraph} problem using recursive understanding, a technique introduced by~\cite{grohe2011finding} and recently used successfully for various parameterized problems~\cite{chitnis2016designing,golovach2020finding,cygan2014minimum,lokshtanov2018reducing}. Specifically,~\cite{lokshtanov2018reducing} prove a meta-result stating that if a problem is FPT on highly-connected graphs and expressible in \emph{Counting Monadic Second-Order Logic}, it is also FPT on general graphs. Thereby, this theorem allows to shortcut the analysis and algorithm for the breakable case of recursive understanding. However, it comes at the price of being nonconstructive and not giving a concrete bound on the runtime. For this reason, it is still valuable to apply recursive understanding directly. In future work, it could be promising to apply and generalize this theorem for directed graphs.
We visualize the overall structure of the algorithm in \Cref{fig:recursive_calls}. 

\begin{figure}[t]
  \centering
  \includegraphics[width=0.6\textwidth,page=2]{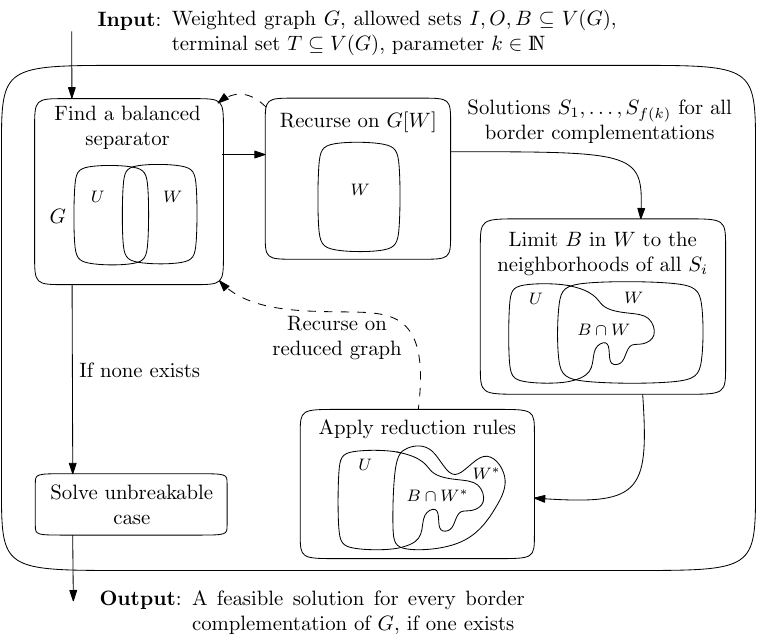}
  \caption{
  An illustration of the general recursive understanding algorithm used in \Cref{sec:scc}.
  There are two recursive calls in total, highlighted with dashed arrows.
  As defined later, only vertices inside $B$ are allowed to be in the neighborhood of a solution.
  $W$ is chosen to be the side of the separation with a smaller intersection with $T$.
  }
  \label{fig:recursive_calls}
\end{figure}

On a high level, recursive understanding algorithms work by first finding a small balanced separator of the underlying undirected graph. If no suitable balanced separator exists, the graph must be highly connected, which makes the problem simpler to solve. In the other case, we reduce and simplify one side of the separator while making sure to keep an equivalent set of solutions in the whole graph. By choosing parameters in the right way, this process reduces one side of the separator enough to invalidate the balance of the separator. Therefore, we have made progress and can iterate with another balanced separator or reach the base case.

In our case, looking for a separator of size at most $k$ makes the framework applicable. Crucially, this is because in any secluded subgraph $G[S]$, where $S \subseteq \ve{G}$, the neighborhood $\nei{S}$ acts as a separator between $S$ and $\ve{G} \setminus \cnei{S}$. Therefore, if no balanced separator of size at most $k$ exists, we can deduce that either $S$ itself or $\ve{G} \setminus S$ must have a small size. This observation makes the problem significantly easier to solve in this case, using the color coding technique developed in~\cite{chitnis2016designing}.

In the other case, we can separate our graph into two balanced parts, $U$ and $W$, with a separator $P$ of size $\abs{P} \le k$. Now, our goal is to solve the same problem recursively for one of the sides, say $W$ and replace that side with an equivalent graph of bounded size that only contains all necessary solutions.
However, finding subsets of solutions is not the same as finding solutions; the solution $S$ for the whole graph could heavily depend on including some vertices in $U$. That being said, the different options for this influence are limited. At most $2^k$ different subsets $X$ of $P$ could be part of the solution. For any such $X$, the solution can only interact across $P$ in a limited number of ways. For finding strongly connected subgraphs, we have to consider for which pair $(x_1, x_2) \in X \times X$ there already is a $x_1$-$x_2$-path in the $U$-part of the complete solution $S$. 
This allows us to construct a new instance for every such possibility by encoding $X$ and the existing paths into $W$. These instances are called \emph{boundary complementations}. We visualize the idea of this construction in \Cref{fig:scc_bc}.

Fundamentally, we prove that an optimal solution for the original graph exists, that coincides in $W$ with an optimal solution to the boundary complementation graph in which $U$ is replaced. Hence, we restrict the space of solutions to only those whose neighborhood in $W$ coincides with the neighborhood of an optimal solution to some boundary complementation. The restricted instance consists of a bounded-size set $B$ of vertices that could be part of $\nei{S}$ and components in $W \setminus B$ that can only be included in $S$ completely or not at all. We introduce graph \emph{extensions} to formalize when exactly these components play the same role in a strongly connected subgraph. Equivalent extensions can then be merged and compressed into equivalent extensions of bounded size.
In total, this guarantees that $W$ is shrunk enough to invalidate the previous balanced separator, and we can restart the process.

\subparagraph*{$\mathcal{F}$-Free Subgraph.}
For any family of graphs $\mathcal{F}$, a graph $G$ is called $\mathcal{F}$-Free if it does not contain any graph in $\mathcal{F}$ as an induced subgraph. Depending on the context both $\mathcal{F}$ and $G$ are either both undirected or both directed.
The widely-studied \textsc{Secluded $\mathcal{F}$-Free Subgraph} problem on undirected graphs is FPT (with parameter $k$) using recursive understanding~\cite{golovach2020finding} or branching on important separators~\cite{jansen2023single} when we restrict it to connected solutions. We study the directed version of the problem and surprisingly, it turns out to be W[1]-hard for almost all forbidden graph families $\mathcal{F}$ even with respect to the parameter $k+w$. Precisely, we prove the following theorem.

\begin{restatable}[]{theorem}{restateffree}
\label{thm:f-free-hard-always}
    Let $\mathcal{F}$ be a non-empty set of directed graphs such that no $F \in \mathcal{F}$ is a subgraph of an inward star. Then, \textsc{Out-Secluded $\mathcal{F}$-Free Weakly Connected Subgraph} is W[1]-hard with respect to the parameter $k+w$ for unit weights.
\end{restatable}

We establish an almost complete dichotomy that highlights the few cases of families~$\mathcal{F}$ for which the problem remains tractable. One of these exceptions is if~$\mathcal{F}$ contains an edgeless graph of any size, where we employ our algorithm for \textsc{Out-Secluded $\alpha$-Bounded Subgraph}. \Cref{thm:f-free-hard-always} also implies the following result for the directed variant of \textsc{Secluded Tree} problem.

\begin{restatable}[]{corollary}{restatedag}
\label{cor:dag_out_k+t}
    \textsc{Out-Secluded Weakly Connected DAG} is W[1]-hard with parameter $k + w$ for unit weights.
\end{restatable}

\subparagraph*{$\alpha$-Bounded Subgraph \& Clique.}

In the undirected setting, the \textsc{Secluded Clique} problem is natural and has been studied specifically. There is an FPT-algorithm running in time $2^{\bigO{k \log k}}n^{\bigO{1}}$ through contracting twins~\cite{golovach2020finding}.
The previous best algorithm however uses the general result for finding secluded $\mathcal{F}$-free subgraphs in time $2^{\bigO{k}}n^{\bigO{1}}$~\cite{jansen2023single}. By using important separators, they require time at least $4^kn^{\bigO{1}}$. 
This property is naturally generalizable to directed graphs via tournament graphs. We go one step further.

The independence number of an undirected or directed graph $G$ is the size of the maximum independent set in $G$ (or its underlying undirected graph).
If $G$ has independence number at most $\alpha$, we also call it \emph{$\alpha$-bounded}. This concept has been used to leverage parameterized results from the simpler tournament graphs to the larger graph class of $\alpha$-bounded graphs~\cite{sahu2023kernelization,fradkin2015edge,misra2023sub}. We prove the following results:

\begin{restatable}[]{theorem}{restateab}
\label{thm:alpha_bounded_fpt}
  \textsc{Out-Secluded $\alpha$-Bounded Subgraph} is solvable in time $(2\alpha + 2)^kn^{\alpha+\bigO{1}}$.
\end{restatable} 

\begin{restatable}[]{theorem}{restateabtotal}
\label{thm:alpha_bounded_total}
  \textsc{Total-Secluded $\alpha$-Bounded Subgraph} is solvable in time $(\alpha + 1)^kn^{\alpha+\bigO{1}}$.
\end{restatable} 

We achieve the goal via a branching algorithm, solving \textsc{Secluded $\alpha$-Bounded Subgraph} for all neighborhood definitions in FPT time (\Cref{thm:alpha_bounded_fpt,thm:alpha_bounded_total}).
Our algorithm initially picks a vertex subset $U \subseteq \ve{G}$ and looks only for solutions in the two-hop neighborhood of $U$. A structural property of $\alpha$-bounded graphs guarantees that any optimal solution is found in this way. 
On a high level, the remaining algorithm depends on two branching strategies. First, we branch on forbidden structures in the two-hop neighborhood of $U$, to ensure that it becomes $\alpha$-bounded. Second, we branch on farther away vertices to reach a secluded set. 

Note that \Cref{thm:alpha_bounded_total} is inherently also an undirected result. Furthermore, the ideas behind the algorithm can in turn be used for the simpler undirected \textsc{Secluded Clique} problem.
By a closer analysis of these two high-level rules, we arrive at a branching vector of $(1,2)$ for \textsc{Secluded Clique}. This results in the following runtime, a drastic improvement on the previous barrier of $4^kn^{\bigO{1}}$.

\begin{restatable}[]{theorem}{restateclique}
\label{thm:clique_better}
  \textsc{Secluded Clique} is solvable in time $1.6181^k n^{\bigO{1}}$.
\end{restatable}

\subparagraph*{Organization.}

We consider \textsc{Total-Secluded Strongly Connected Subgraph} and prove \Cref{cor:scc_algo} in \Cref{sec:scc}.
The hardness result about \textsc{Out-Secluded $\mathcal{F}$ Weakly Connected Subgraph} in \Cref{thm:f-free-hard-always} is proved in \Cref{chap:hardness}.
In \Cref{chap:tour}, we give the algorithms and proofs for \Cref{thm:alpha_bounded_fpt,thm:alpha_bounded_total,thm:clique_better}.

\subparagraph*{Notation.}

Let $G$ be a directed graph. For a vertex $v \in \ve{G}$, we denote the \emph{out-neighborhood} by $\outNei{v} = \set{u}{(v,u) \in \e{G}}$ and the \emph{in-neighborhood} by $\inNei{v} = \set{u}{(u,v) \in \e{G}}$.
The \emph{total-neighborhood} is defined as $\nei{v} = \outNei{v} \cup \inNei{v}$.
We use the same notation for sets of vertices $S \subseteq \ve{G}$ as $\outNei{S} = \bigcup_{v \in S} \outNei{v} \setminus S$. 
Furthermore, for all definitions, we also consider their \emph{closed} version that includes the vertex or vertex set itself, denoted by $\coutNei{v} = \outNei{v} \cup \set{v}$.
For a vertex set $S \subseteq \ve{G}$, we write $\induced{G}{S}$ for the subgraph induced by $S$ or $G - S$ for the subgraph induced by $\ve{G} \setminus S$. We also use $G - v$ instead of $G - \set{v}$.

When we refer to a component of a directed graph, we mean a component of the underlying undirected graph, that is, a maximal set that induces a weakly connected subgraph. In contrast, a strongly connected component refers to a maximal set that induces a strongly connected subgraph.
For standard parameterized definitions, we refer to~\cite{cygan2015parameterized}.

\zzcommand{\scs}{\textsc{TSSCS}}

\section{Total-Secluded Strongly Connected Subgraph}
\label{sec:scc}
In this section, we investigate the \textsc{Total-Secluded Strongly Connected Subgraph} problem, or \scs{} for short. First, we prove that the problem is NP-hard in general graphs, motivating analysis of its parameterized complexity.

\begin{theorem}
\label{thm:total_scc_np_hard}
  \scs{} is NP-hard, even for unit weights.
\end{theorem}
\begin{proof}
\begin{figure}
  \centering
  \hfill
  \begin{subfigure}{0.49\textwidth}
    \centering
    \includegraphics[width=0.6\textwidth,page=2]{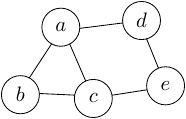}
    \caption{An undirected graph $G$ with maximum clique size 3.}
  \end{subfigure}
  \hfill
  \begin{subfigure}{0.49\textwidth}
    \centering
    \includegraphics[width=0.70\textwidth,page=3]{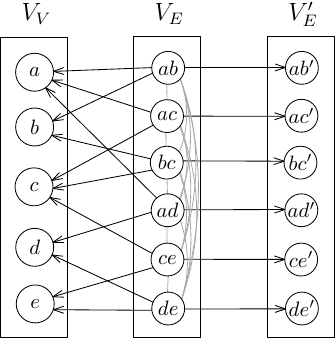}
    \caption{The \scs{} instance $G'$ created by our reduction in the proof of \Cref{thm:total_scc_np_hard}. }
  \end{subfigure}
  \hfill
  \caption{A visualization of the reduction in the proof of \Cref{thm:total_scc_np_hard}.  }
  \label{fig:clique_reduction_scc}
\end{figure}

  We reduce from the \textsc{Clique} problem which is NP-hard, where given a graph the problem is to check if there is a clique of size at least $k$. Given an instance $(G,k)$ with $k\ge 2$, we reduce it to an instance of \scs{} $(G',w,k')$ as follows:
  
  We create the graph $G'$ with $\ve{G'} \coloneqq \ve{G}\cup V_E \cup V_E'$, where $V_E \coloneqq \set{v_e}{e\in E(G)}$ and $V_E'\coloneqq \set{v_e'}{e\in E(G)}$. For any vertex $x\in V(G')$, the weight of $v$ is 1.
  The new graph has edges \[
  \e{G'} \coloneqq \set{(v_{e_1}, v_{e_2})}{v_{e_1} \ne v_{e_2} \in V_E} \cup \set{(v_e, v_e'), (v_e, a), (v_e, b)}{e = \set{a,b} \in \e{G}}.
  \] 
    We set $k' = k + \abs{\e{G}}$ and $w = \binom{k}{2}$. For an illustration of the construction of $G'$, refer to \Cref{fig:clique_reduction_scc}.

 Next, we prove that $G$ has a clique of size $k$ if and only if $G'$ has a strongly connected subgraph $H$ of weight at least $w$ with the size of total neighborhood of $H$ at most $k'$. For the forward direction, let $C \subseteq \ve{G}$ be a clique of size $k$ in $G$. Then, we can choose $S = \{v_{\{a, b\}} \mid a, b \in C\}$. Since $C$ is a clique, $S \subseteq V_E$. Also, $S$ has weight $w$ since a clique of size $k$ contains $\binom{k}{2}$ edges. Furthermore, the neighborhood of $S$ consists of all $v_e'$ for $v_e \in S$, all $v_e \notin S$, and all $v_a$ for $a \in C$, giving a total size of $k'$. Therefore, $S$ is a solution for \scs{}.

  Let $S$ be a solution for \scs{} in $G'$. For $k \ge 2$, the solution must include at least one vertex from $v_E$ to reach the desired weight. Therefore, no vertex from $\ve{G'} \cup V_E'$ can be included to not violate connectivity. The size of the neighborhood of $S$ will then be $\abs{\e{G}}$ increased by the number of incident vertices to edge-vertices picked in $S$. Since $S$ must be an edge set of size at least $\binom{k}{2}$ with at most $k$ incident vertices, this must induce a clique of size at least $k$ in $G$.
\end{proof}

The proof of \Cref{thm:total_scc_np_hard} also shows that \scs{} is W[1]-hard when parameterized by $w$, since $w$ in the proof also only depends on the parameter for \textsc{Clique}.
In the following subsections, we describe the recursive understanding algorithm to solve \scs{} parameterized by $k$.
We follow the framework by~\cite{chitnis2016designing,golovach2020finding} and first introduce generalized problems in \Cref{sec:scc_bc}. In \Cref{sec:unbreak}, we solve the case of unbreakable graphs. We introduce graph \emph{extensions} in \Cref{sec:scc_extensions} as a framework to formulate our reduction rules and full algorithm in \Cref{sec:solving_scc}.

\subsection{Boundaries and boundary complementations}\label{sec:scc_bc}

In this subsection, we first define an additional optimization problem that is useful for recursion. Then, we describe a problem-specific \emph{boundary complementation}. Finally, we define the auxiliary problem that our algorithm solves, which includes solving many similar instances from the optimization problem.

\zzcommand{\scsrec}{\textsc{Max \scs{}}}
\begin{tcolorbox}[enhanced,title={\color{black} {\scsrec{}}}, colback=white, boxrule=0.4pt,
	attach boxed title to top left={xshift=.3cm, yshift*=-2.5mm},
	boxed title style={size=small,frame hidden,colback=white}]
	
	\textbf{Input:}  
  A directed graph $G$, subsets $I,O,B \subseteq \ve{G}$, a weight function $\wOp \colon \ve{G} \to \N$, and an integer $k \in \N$

	\textbf{Output:} A set $S \subseteq \ve{G}$ that maximizes $\wOp(S)$ subject to $I \subseteq S$, $O \cap S = \emptyset$, $N(S) \subseteq B$, $\abs{\nei{S}} \le k$, and $\induced{G}{S}$ strongly connected, or report that no feasible solution exists. 
\end{tcolorbox} 

Note that this problem generalizes the optimization variant of \scs{} by setting $I \coloneqq O \coloneqq \emptyset$ and $B \coloneqq \ve{G}$. However, \scsrec{} allows us to put more constraints on recursive calls, enforcing vertices to be included or excluded from the solution and neighborhood.

\begin{figure}[t]
    \centering
    \hfill
    \begin{subfigure}{0.48\textwidth}
      \centering
      \includegraphics[width=0.9\textwidth,page=2]{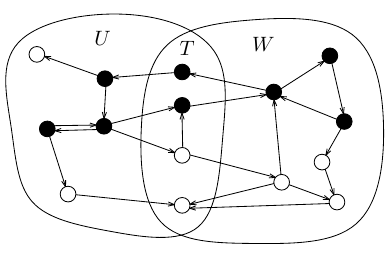}
      \caption{A strongly connected subgraph $S$ of a graph $G$ with $k=5$ neighbors. Black vertices are part of $S$.}
    \end{subfigure}
    \hfill
    \begin{subfigure}{0.48\textwidth}
      \centering
      \includegraphics[width=0.9\textwidth,page=3]{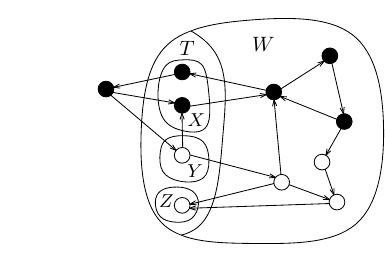}
      \caption{The boundary complementation that admits an equivalent feasible solution when setting $k' \coloneqq k - 2$.}\label{fig:scc_bcb}
    \end{subfigure}
    \hfill
    \caption{A visualization of a solution in the original graph and a solution in a boundary complementation. Every partial solution in $U$ can be represented by a boundary complementation.}\label{fig:scc_bc}
\end{figure}

\begin{definition}[Boundary Complementation]\label{def:scc_border_complementation}
  Let $\mathcal{I} = (G,I,O,B,\wOp,k)$ be a \scsrec{} instance. Let $T \subseteq \ve{G}$ be a set of \emph{boundary terminals} with a partition $X,Y,Z \subseteq T$ and let $R \subseteq X \times X$ be a relation on $X$. Then, we call the instance $(G',I',O',B,\wOp',k')$ a \emph{boundary complementation} of $\mathcal{I}$ and $T$ if
   \begin{enumerate}
    \item $G'$ is obtained from $G$ by adding vertices $u_{(a,b)}$ for every $(a,b) \in R$ and edges $(a, u_{(a,b)})$, $(u_{(a,b)}, b)$, and for every $y \in Y$ additionally $(u_{(a,b)}, y)$,
    \item $I' \coloneqq I \cup X \cup \{u_{(a,b)} \mid (a,b) \in R\}$,
    \item $O' \coloneqq O \cup Y \cup Z$,
    \item $\wOp'(v) \coloneqq \w{v}$ for $v \in \ve{G}$ and $\wOp'(u_{(a,b)}) \coloneqq 0$ for $(a,b) \in R$, and
    \item $k' \le k$.\qedhere
  \end{enumerate}
\end{definition}

See \Cref{fig:scc_bc} for an example boundary complementation.
The intuition here should be that if we take the union of $G$ with any other graph $H$ and only connect $H$ to $G$ at the vertices in $T$, then $(X,Y,Z,R)$ encodes all possibilities of how a solution in $G \cup H$ could behave from $G$'s point of view. So, for any solution $S$ to \scsrec{} in $G \cup H$, there is some boundary complementation for $G$ in which we can solve and exchange $S \cap G$ for that solution. Later, we prove a statement that is similar to this intuition.

To employ recursive understanding, we need a boundaried version of the problem. Intuitively, this problem is the same as the previous \scsrec{} but for a small part of the graph we want to try out every possibility, giving many very similar instances. This small part will later represent a separator to a different part of the graph. 

\zzcommand{\scsborder}{\textsc{Boundaried \scsrec{}}}
\begin{tcolorbox}[enhanced,title={\color{black} {\scsborder{}}}, colback=white, boxrule=0.4pt,
	attach boxed title to top left={xshift=.3cm, yshift*=-2.5mm},
	boxed title style={size=small,frame hidden,colback=white}]
	
	\textbf{Input:}  
A \scsrec{} instance $\mathcal{I} = (G,I,O,B,\wOp,k)$ and a set of boundary terminals $T \subseteq \ve{G}$ with $\abs{T} \le 2k$

	\textbf{Output:}
A solution to \scsrec{} for each boundary complementation $\mathcal{I}'$ of $\mathcal{I}$ and $T$, or report that no solution exists.
\end{tcolorbox}

To even have a chance to solve this problem, we need to make sure that there are not too many boundary complementations. The following lemma bounds that number in terms of $k$. 

\begin{lemma}
\label{lem:number_border_complementations}
  For a \scsrec{} instance $(G,I,O,B,\wOp,k)$ and $T \subseteq \ve{G}$, there are at most $3^{\abs{T}}2^{\abs{T}^2}(k+1)$ many boundary complementations, which can be enumerated in time $2^{\bigO{\abs{T}^2}}n^{\bigO{1}}$.
\end{lemma}
\begin{proof}
Every element of $T$ has to be in either $X$, $Y$, or $Z$, which gives $3^{\abs{T}}$ possible arrangements. Every one of the $\abs{X}^2$ elements in $X\times X$ can either be in $R$ or not, giving $2^{\abs{X}^2} \le 2^{\abs{T}^2}$ possible arrangements. Furthermore, there are $k+1$ choices for $0 \le k' \le k$. Multiplying these gives the final number. By enumerating all the respective subsets and constructing $G'$ in time $n^{\bigO{1}}$, we can enumerate all boundary complementations in the claimed time.
\end{proof}

\subsection{Unbreakable Case}\label{sec:unbreak}

This subsection gives the algorithm for the base case of our final recursive algorithm, when no balanced separator exists. We start by giving the definitions of separations and unbreakability.

\begin{definition}[Separation]
  Given two sets $A, B \subseteq \ve{G}$ with $A \cup B = \ve{G}$, we say that $(A, B)$ is a \emph{separation of order $\abs{A \cap B}$} if there is no edge with one endpoint in $A \setminus B$ and the other endpoint in $B \setminus A$.
\end{definition}

\begin{definition}[Unbreakability]
  Let $q,k \in \N$. An undirected graph $G$ is \emph{$(q,k)$-unbreakable} if for every separation $(A,B)$ of $G$ of order at most $k$, we have $\abs{A \setminus B} \le q$ or $\abs{B \setminus A} \le q$.
\end{definition}

The next lemma formalizes that the neighborhood of any solution gives you a separator of order $k$. If the graph is unbreakable, either the solution or everything but the solution must be small. With this insight, the statement should be intuitive. We only need to fill in the details since the graph changes slightly when considering the boundary complementation. However, the proof is identical to~\cite[Lemma~12]{golovach2020finding}, so we omit it.

\begin{lemma}\label{lem:unbreak_small_or_large}
  Let $\mathcal{I}$ be a \scsborder{} instance on a $(q,k)$-unbreakable graph $G$. Then, for each set $S$ in a solution of $\mathcal{I}$, either $\abs{S \cap \ve{G}} \le q$ or $\abs{\ve{G} \setminus S} \le q + k$.
\end{lemma}

\begin{lemma}[\protect{\cite[Lemma~1.1]{chitnis2016designing}}]\label{lem:find_sets}
  Given a set $U$ of size $n$ and integers $a,b \in \N$, we can construct in time $2^{\bigO{\min\set{a,b}\log (a+b)}}n\log n$ a family $\mathcal{F}$ of at most $2^{\bigO{\min\set{a,b}\log (a+b)}}\log n$ subsets of $U$ such that the following holds. For any sets $A, B \subseteq U, A \cap B = \emptyset, \abs{A} \le a, \abs{B} \le b$, there is a set $S \in \mathcal{F}$ with $A \subseteq S$ and $B \cap S = \emptyset$.
\end{lemma}


\begin{theorem}
\label{thm:unbreakable_scc}
  \scsborder{} on $(q,k)$-unbreakable graphs can be solved in time $2^{\bigO{k^2 \log(q)}}n^{\bigO{1}}$.
\end{theorem}
\begin{proof}
  Our algorithms starts by enumerating all boundary complementations of an instance $\mathcal{I}$. For a \scsrec{} instance $\mathcal{I'} = (G,I,O,B,\wOp,k)$, we know by \Cref{lem:unbreak_small_or_large} that the solution must be either of size at most $q+4k^2$ including the newly added vertices $u_{(a,b)}$, or at least $\abs{\ve{G}} - (q+k+4k^2)$. Let $s = q+k+4k^2 \ge q + 4k^2$.

  This allows us to address the two possible cases for the solution size separately. We give one algorithm to find the maximum weight solution of size at most $s$ and one algorithm to find the maximum weight solution of size at least $\abs{\ve{G}} - s$. In the end, we return the maximum weight of the two, or none if both do not exist. 

  \subparagraph*{Finding a small solution} Our algorithm works as follows. 
  \begin{enumerate}
    \item Apply the algorithm from \Cref{lem:find_sets} with $U = \ve{G}, a = s, b = k$ to compute a family $\mathcal{F}$ of subsets of $\ve{G}$.

    \item For every $F \in \mathcal{F}$, consider the strong components of $\induced{G}{F}$ separately.

    \item For a strong component $Q$, we check if $Q$ is a feasible solution and return the maximum weight one.
  \end{enumerate}

  \subparagraph*{Finding a large solution} For this case, our algorithm looks as follows. 
  \begin{enumerate}
    \item Compute the strongly connected components of $G$.

    \item For every strongly connected component $C$, we construct the graph $G_C$ by taking $\induced{G}{\cnei{C}}$ and adding a single vertex $c$ with edges $(c,v)$, for every $v \in \nei{C}$. 

    \item Then, run the algorithm from \Cref{lem:find_sets} with $U = \ve{G_C}, a = s+1, b = k$ to receive a family $\mathcal{F}$ of subsets of $\ve{G_C}$.

    \item For every $F \in \mathcal{F}$, find the component including $c$.

    \item \label[step]{it:5}For each such component $Q$, check if $\ve{G_c} \setminus \cnei{Q}$ is a feasible solution and return the maximum weight one.
  \end{enumerate}

  \subparagraph*{Correctness}
  By \Cref{lem:find_sets}, for any small solution $S$ there is $F \in \mathcal{F}$ with $S \subseteq F$ and $\nei{S} \cap F = \emptyset$, so $S$ must be both a component that is also strongly connected of $\induced{G}{F}$. Therefore, we enumerate a superset of all solutions of size at most $s$ and we find the maximum weight small solution.

  For the large solution, we claim that we find a maximum weight solution for this case. Clearly, every strongly connected subgraph of $G$ must be a subgraph of a strong component of $G$. Let $S'$ be a large solution that is a subset of a strong component $C$. Consider the set $S = \ve{G_C} \setminus \cnei{S'}$ in $G_C$. Then, we must have $\abs{S} \le a$ since $S'$ is large. Also, $\nei{S} \subseteq \nei{S'}$ by definition and therefore $\abs{\nei{S}} \le b$. Thus, $S$ is considered in \Cref{it:5} if it is weakly connected.

  If there is a $v \in S$ that is not in the same component as $c$, we take the component of $v$ in $S$ and include it in $S'$. Then, $S'$ is still strongly connected but has a strictly smaller neighborhood size and equal or greater weight since weights are non-negative. We can repeat this procedure until $S$ is a single component and will be enumerated by the algorithm. Therefore, our algorithm finds a solution of weight at least $\w{S'}$.

  \subparagraph*{Total Runtime} Both cases make use of at most $n$ calls to the algorithm from \Cref{lem:find_sets} with some small modifications. For every returned sets, both algorithms compute the components and verify strong connectivity. Therefore, we can bound the runtime per boundary complementation by $2^{\bigO{\min\set{s,k}\log(s+k)}}n^{\bigO{1}} = 2^{\bigO{k\log(q+k)}}n^{\bigO{1}}$. By \Cref{lem:number_border_complementations}, enumerating all boundary complementations adds a factor of $2^{\bigO{k^2}}n^{\bigO{1}}$, which gives the desired runtime.
\end{proof}

\subsection{Compressing Graph Extensions}\label{sec:scc_extensions}

Before we give the complete algorithm, we define a routine that compresses a part of the graph. We aim for two crucial properties in the compressed part. First, the part after compression should be equivalent to the part before compression in terms of which strongly connected components can be formed. Second, we want the size of the compressed part to be functionally bounded by the size of remaining graph.

To achieve this goal, we first formally define sufficient properties to reason about this equivalence and bound the number of equivalence classes.
First, we define the notion of a graph \emph{extension}, a way to extend one graph with another. The extension will play the role of the compressed part of the graph. This concept allows us to speak more directly about graph properties before and after exchanging a part of the graph with a different one.

\begin{definition}[Extension]
Given a directed graph $G$, we call a pair $\dexpair$ an \emph{extension of} $G$ if $D$ is a directed graph and $\dexset \subseteq (\ve{G} \times \ve{D}) \cup (\ve{D} \times \ve{G})$ is a set of pairs between $G$ and $D$. 
We name the graph $\ex{G}{D}{\dexset} \coloneqq (\ve{G} \cup \ve{D}, \e{G} \cup \e{D} \cup \dexset)$, that can be created from the extension, $G$ \emph{extended by} $\dexpair$.
\end{definition}

We use extensions to construct extended graphs.
Intuitively, an extension of $G$ is a second graph $D$ together with an instruction $\dexset$ on how to connect $D$ to $G$.

Next, we identify three important attributes of extensions in our context. Later, we show that these give a sufficient condition on when two extensions form the same strongly connected subgraphs.
  For this, consider a directed graph $G$ with an extension $\dexpair$. For $U \subseteq \ve{D}$, we write $\inNeiG{\dexset}{U}$ as a shorthand for $\inNeiG{\ex{G}{D}{\dexset}}{U}$, that is, all $v \in \ve{G}$ with $(v,u) \in \dexset$ for some $u \in U$. Define $\outNeiG{\dexset}{v}$ analogously.
  Write $\scc{D}$ for the \emph{condensation} of $D$, where every strongly connected component $C$ of $D$ is contracted into a single vertex.
  Define $\sourcee{D}{\dexset}, \sinke{D}{\dexset} \subseteq 2^{\ve{G}}$ such that 
  \begin{align*}
    \sourcee{D}{\dexset} &\coloneqq \set{\inNeiG{\dexset}{U}}{U \subseteq \ve{D} \text{ is a source component in } \scc{D}} \text{ and}\\
    \sinke{D}{\dexset} &\coloneqq \set{\outNeiG{\dexset}{U}}{U \subseteq \ve{D} \text{ is a sink component in } \scc{D}},
  \end{align*}
  that is, for every strongly connected source component $C$ in $\scc{D}$, $\sourcee{D}{\dexset}$ contains the set of all $v \in \ve{G}$ such that $\died{v}{u} \in \dexset$ for some $u \in C$ and analogously for $\sinke{D}{\dexset}$.
  Furthermore, define \[\conne{D}{\dexset} \coloneqq \set{(a,b) \in \ve{G}^2}{\text{there is a $d_1$-$d_2$-path in $D$ with } (a,d_1), (d_2,b) \in \dexset}, \] that is, all $(a,b)$ such that there is an $a$-$b$-path in $\ex{G}{D}{\dexset}$, whose intermediate vertices and edges belongs to D. Refer to \Cref{fig:extension_compression} for examples of extensions and the three sets.

\begin{definition}[Equivalent Extensions]
    Let $G$ be a directed graph. We say that two extensions $\expairi{1}$ and $\expairi{2}$ of $G$ are \emph{equivalent} if \[(\source{D_1,E_{GD_1}}, \sink{D_1,E_{GD_1}}, \conn{D_1,E_{GD_1}}) = (\source{D_2,E_{GD_2}}, \sink{D_2,E_{GD_2}}, \conn{D_2,E_{GD_2}}).\]
\end{definition}

Clearly, extension equivalence defines an equivalence relation.
The next statement reveals the motivation behind the definition of extension equivalence. It gives us a sufficient condition for two extensions being exchangeable in a strongly connected subgraph.

\begin{lemma}
\label{lem:source_sink_conn_equiv}
  Let $G$ be a directed graph with two equivalent extensions $\expairi{1}$ and $\expairi{2}$. Let $U \subseteq \ve{G}$ be nonempty such that the extended graph $\ex{\induced{G}{U}}{D_1}{\exset{\induced{G}{U}}{D_1}}$ is strongly connected. Then $\ex{\induced{G}{U}}{D_2}{\exset{\induced{G}{U}}{D_2}}$ is also strongly connected.
\end{lemma}
\begin{proof}
  We construct a $v_1$-$v_2$-path for all $v_1, v_2 \in U \cup \ve{D_2}$ that only uses edges in $\e{D_2}$, $\exset{G}{D_2}$, and $\induced{G}{U}$ by case distinction.

\begin{description}
    \item[Paths $U \to U$.]  Let $u_1, u_2 \in U$. If there is a path from $u_1$ to $u_2$ in $\induced{G}{U}$, this path also exists after exchanging $\expair{G}{D_1}$ to $\expair{G}{D_2}$. If the path passes through $D_1$, since $\conn{D_1,E_{GD_1}} = \conn{D_2,E_{GD_2}}$, we can exchange all subpaths through $D_1$ by subpaths through $D_2$.

    \item[Paths $\ve{D_2} \to U$.] Let $v \in \ve{D_2}, u \in U$. We construct a $v$-$u$-path by first walking from $v$ to any sink component $T$ in $\scc{D_2}$. If there is no edge $(t,u') \in \exset{G}{D_2}$ with $t \in T, u' \in U$ that we can append, since $\sink{D_1,\expair{G}{D_1}} = \sink{D_2,\expair{G}{D_2}}$, there must also be a sink component in $\scc{D_1}$ with no outgoing edge to $U$. However, this is a contradiction to the fact that $\ex{\induced{G}{U}}{D_1}{\exset{\induced{G}{U}}{D_1}}$ is strongly connected with nonempty $U$. Therefore, we can find a $(t,u')$ to append for some $t \in T, u' \in U$. From $u'$, there is already a path to $u$, as proven in the first case.

    \item[Paths $U \to \ve{D_2}$.] Next, we construct a $u$-$v$-path backwards by walking from $v$ backwards to a source $s$ in $D_2$. Analogously, there is an edge $(u',s) \in \exset{G}{D_2}$ for some $u' \in U$ since $\source{D_1,E_{GD_1}} = \source{D_2,E_{GD_2}}$, which we append. From $u$, there is a path to $u'$, as proven in the first case, which we prepend to the rest of the path.

    \item [Paths $\ve{D_2} \to \ve{D_2}$.] Let $v_1, v_2 \in \ve{D_2}$. To construct a $v_1$-$v_2$-path, we can just walk from $v_1$ to any $u \in U$ and from there to $v_2$ as shown before.\qedhere
\end{description}
\end{proof}

Furthermore, observe that the union of two extensions creates another extension where source, sink and connection sets correspond exactly to the union of the previous sets. Hence, the union of two equivalent extensions will again be equivalent. This fact is formalized in the next observation and  will turn out useful in later reduction rules.

\begin{figure}
    \centering
    \includegraphics[width=0.65\linewidth]{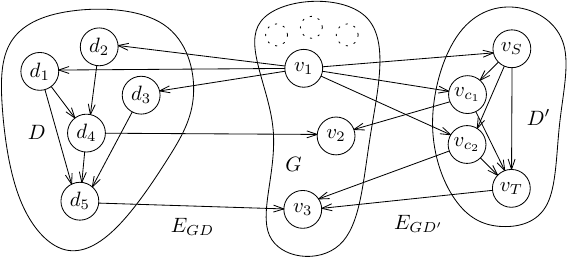}
    \caption{Two example extensions of a graph $G$. Observe that $\source{D,\dexset} = \set{\set{v_1}}$, $\sink{D,\dexset} = \set{\set{v_3}}$, and $\conn{D,\dexset} = \set{(v_1,v_2), (v_1, v_3)}$. The extension $(D', E_{GD'})$ not only has the same sets $\sourceOp, \sinkOp, \connOp$ and is thereby equivalent; it is also the compressed extension of $\dexpair$. Since all sources $d_1,d_2,d_3$ have the same in-neighborhood, they are represented by the single vertex $v_S$.}
    \label{fig:extension_compression}
\end{figure}

\begin{observation} \label{lem:union_equiv_if_equiv}
  Let $G$ be a directed graph with two equivalent extensions $\expairi{1}$ and $\expairi{2}$. Consider the extension defined by $D \coloneqq (\ve{D_1} \cup \ve{D_2}, \e{D_1} \cup \e{D_2})$ and $\exset{G}{D} \coloneqq \exset{G}{D_1} \cup \exset{G}{D_2}$. Then $\dexpair$ is equivalent to $\expairi{1}$ and $\expairi{2}$.
\end{observation}

Now, we finally define our compression routine, which compresses an extension to a bounded size equivalent extension.
If an extension is strongly connected, it is easy to convince yourself that it is always possible to compress the extension to a single vertex. Otherwise, we add one source vertex per neighborhood set in $\source{D,\dexset}$ as well as one sink vertex per neighborhood set in $\sink{D,\dexset}$, realizing the same $\sourceOp$ and $\sinkOp$. Then, we add vertices in between suitable source and sink vertices to realize exactly the same connections in $\connOp$ without creating additional ones. 
The result of a compression is visualized in \Cref{fig:extension_compression}. Now, we describe the procedure formally. 

\begin{figure}[t]
  \begin{minipage}[c]{0.64\linewidth}
      \centering
      \begin{subfigure}{0.49\textwidth}
        \includegraphics[width=\textwidth,page=1]{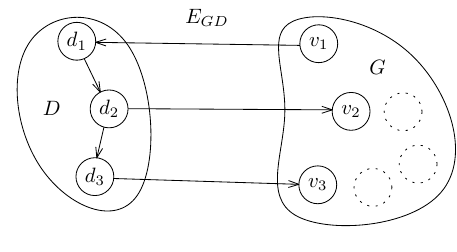}
        \caption{An extension $\dexpair$}\label{fig:scc_comp1}
      \end{subfigure}
      \hfill
      \begin{subfigure}{0.49\textwidth}
        \includegraphics[width=\textwidth,page=2]{figures/scc_comp}
        \caption{$\dcompG$}\label{fig:scc_comp2}
      \end{subfigure}
      \caption{An example for compressing an extension.
      While the size of the extension increases in this example, in general, the size of a compressed extension can be bounded by $\abs{\ve{G}}$.}\label{fig:scc_comp}
  \end{minipage}
  \hfill
  \begin{minipage}[c]{0.34\linewidth}
      \centering
      \includegraphics[width=0.8\textwidth]{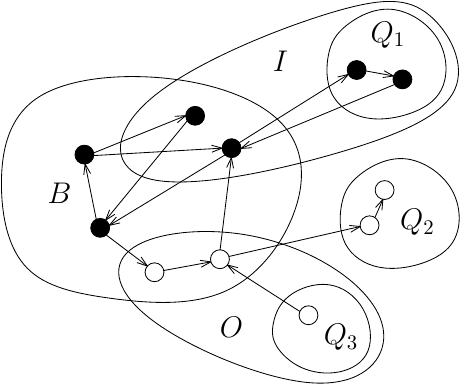}
      \caption{A visualization of how the sets $I$, $O$, and $B$ can overlap after the application of \Cref{red:in_out}. Any component in $G-B$ can be in $I$, $O$, or none of them, but not both. The black vertices form a feasible solution.}\label{fig:iob}
  \end{minipage}
\end{figure}

Let $G$ be a directed graph with an extension $\dexpair$. Our \emph{compression routine} returns an extension that we call \emph{compressed extension}, denoted as $\dcompG$. 
\subparagraph*{Compression Routine}
\begin{itemize}
    \item  If $D$ is strongly connected, we contract $D$ to one vertex $v$ and remove self-loops and multiple edges. We adjust $E_{GD}$ by using $v$ instead of $\ve{D}$ and removing multiple edges.
      
    \item Otherwise,  $\dcompG \coloneqq (D',\exset{G}{D'})$, where $(D',\exset{G}{D'})$ is an extension such that
    \begin{align*}
  \ve{D'} \coloneqq &\set{v_S}{S \in \source{D,\dexset}} \cup \set{v_T}{T \in \sink{D,\dexset}} \cup \set{v_c}{c \in \conn{D,\dexset}},\\
  \exset{G}{D'} \coloneqq &\set{(s,v_S)}{S\in \source{D,\dexset}, s\in S} \cup \set{(v_T,t)}{T\in \sink{D,\dexset}, t\in T }\\
   \cup &\set{(a,v_c),(v_c,b)}{c=(a,b) \in \conn{D,\dexset}}.
    \end{align*}

    To define $\e{D}$, consider every source component $C_s$ and sink component $C_t$ in $\scc{D}$ such that $C_t$ is reachable from $C_s$ in $D$. Let $S \coloneqq \inNeiG{\dexset}{C_s}$ and $T \coloneqq \inNeiG{\dexset}{C_t}$ be the corresponding sets in $\source{D,\dexset}$ and $\sink{D,\dexset}$.
    \begin{itemize}
        \item Add the edge $(v_S, v_T)$ to $\e{D}$.
        \item For every $c = (a,b) \in \conn{D,\dexset}$ that satisfies $(s,b) \in \conn{D,\dexset}$ and $(a,t) \in \conn{D,\dexset}$ for every $s \in S, t \in T$, add the edges $(v_S, v_c)$ and $(v_c, v_T)$ to $\e{D}$.
    \end{itemize}
 \end{itemize}

Now we go on to prove the properties that are maintained while compressing. Then, we bound the size of a compressed extension and thus also the number of equivalence classes.

\begin{lemma}
\label{lem:comp_weakly_conn_equiv}
  Let $G$ be a directed graph with an extension $\dexpair$ and let $(D', E_{GD'})$ be the compressed extension of $\dexpair$. 
  Then, the following are true.
  \begin{enumerate}
      \item If $D$ is weakly connected, then $D'$ is also weakly connected. 
      \item $D$ is strongly connected if and only if $D'$ is strongly connected.
      \item $(D', E_{GD'})$ is equivalent to $\dexpair$. 
  \end{enumerate}
\end{lemma} 
\begin{proof}
  For the first property, assume that $D$ is weakly connected.
  We know by definition that every sink in $D'$ is reached by at least one source.
  Consider a $v_c$ with $c = (a,b) \in \conn{D,\dexset}$. To show that $v_c$ is connected to some $v_S$ and $v_T$, consider the path from $d_1$ to $d_2$ in $D$ that realizes this connection. There must be a source component $C_S$ and a sink component $C_T$ in $\scc{D}$ such that $d_1$ is reachable from $C_S$ and $C_T$ is reachable from $d_2$. Therefore, any vertex in $C_S$ can also reach $d_2$ and $d_1$ can reach every vertex in $C_T$. By definition of compression, these two components ensure that $v_c$ is connected. 
  
  It remains to show that any source is reachable by any other source in the underlying undirected graph. Let $v_S, v_{S'}$ be two sources in $D'$ with corresponding source components $C$, $C'$ in $\scc{D}$. Since $D$ is weakly connected, there is a path from $C$ to $C'$ in the underlying undirected graph. Whenever the undirected path uses an edge in a different direction than the one before, we extend the path to first keep using edges in the same direction until a source or sink component is reached and then go back to the switching point. This new path can directly be transferred to $D'$, where we only keep the vertices corresponding to source and sink components. By definition, this is still a path in $D'$ that connects $v_S$ to $v_{S'}$, and $D'$ is weakly connected

  The second property is simple to verify, since strongly connected graphs are by definition compressed to single vertices. If $D$ is not strongly connected, $D'$ will have at least one source and one sink that are not the same.
  
  Regarding the equivalence, we create one source for every $S \in \source{D,\dexset}$ with the same set of incoming neighbors and create no other sources. Therefore, $\source{D', E_{GD'}} = \source{D,\dexset}$ and $\sink{D',E_{GD'}} = \sink{D,\dexset}$ follows analogously.
  For every connection $c \in \conn{D,\dexset}$, we create $v_c$ in $D'$ that realizes this connection. Therefore, we know that $\conn{D',E_{GD'}} \supseteq \conn{D,\dexset}$. Since $v_c$ is only reachable from sources and reaches only sinks that do not give new connections, we arrive at $\conn{D',E_{GD'}} = \conn{D,\dexset}$. 
\end{proof}

We also bound the number of possible different compression outputs as well as their size.

\begin{lemma}
\label{lem:range_size_comp}
  For a directed graph $G$, there can be at most $2^{2\cdot 2^{\abs{\ve{G}}} + \abs{\ve{G}}^2}$ different compressed extensions.
  Furthermore, every compressed extension has at most $2^{\abs{\ve{G}}+1} + \abs{\ve{G}}^2$ vertices.
\end{lemma}
\begin{proof}
  There are $2^{\abs{\ve{G}}}$ subsets of $\ve{G}$, each of them can be the neighborhood of a source or sink. Additionally, every one of the at most $\abs{\ve{G}}^2$ can form a connection or not. This proves the first claim. 

  Suppose $\compOp_G$ outputs $\dexpair$. For every subset $U \subseteq \ve{G}$, there can be at most one source and one sink in $D$ that has $U$ as outgoing or incoming neighbors. Also, there are at most $\abs{\ve{G}}^2$ pairs of vertices in $G$ that can form a connection. This bounds the number of connection vertices in $D$.
\end{proof}

In the past section, we have defined graph extensions and an equivalence relation on them that captures the role they can play in forming strongly connected subgraphs. We have presented a way to compress an extension such that its size only depends on the size of $G$, while remaining equivalent.
In the next section, we will use this theory to design reduction rules for \scsborder{}. Crucially, we view components outside of the set $B$ as extensions, which allows us to keep only one compressed extension of each equivalence class.

\subsection{Solving \scsborder{}}\label{sec:solving_scc}

We start by giving some reduction rules for a \scsborder{} instance $\mathcal{I} = (G, I, O, B, \wOp,  k, T)$. Additionally, we assume that $T \subseteq B$ to ensure that we do not change $T$ when changing $G - B$. This condition will always be satisfied in our algorithm.

The first reduction rule extends the sets $I$ and $O$ to whole components of $G - B$. This is possible since no solution can include only part of a component without its neighborhood intersecting the component. Remember that components always refer to weakly connected components. 

\begin{reduction*}
\label{red:in_out}
  Let $Q$ be a component of $G - B$. If $Q \cap O \ne \emptyset$, set $O = O \cup \cnei{Q}$. If $\cnei{Q} \cap I \ne \emptyset$, set $I = I \cup Q$. If both cases apply, the instance has no solution.
\end{reduction*} 
\begin{proof}[Proof of Safeness]
  For the first case, assume $Q \cap O \ne \emptyset$. Notice that if any vertex of $\cnei{Q}$ is in the solution, there also has to be some vertex of $Q$ in the neighborhood of the solution since $Q$ is weakly connected. This, however, is not possible as $Q \cap B = \emptyset$.

  Similarly, for the second case, assume $\cnei{Q} \cap I \ne \emptyset$. If any vertex of $Q$ is not in the solution, there has to be a vertex of $Q$ in the neighborhood of the solution since $Q$ is weakly connected, which again is impossible. 
  
  Therefore, we can safely apply the first two cases. Since in the last case $I \cap O \ne \emptyset$, there can clearly be no solution.
\end{proof}

If this reduction rule is no longer applicable, every component in $G-B$ is either completely in $I$, completely in $O$, or intersects with none of the two. 


From this point, we will use extensions from the previous section for components of $G-B$. Namely, for a component $Q$ of $G-B$, let $E_{BQ}$ be all the edges with exactly one endpoint in $B$ and one in $Q$. Then, $(\induced{G}{Q}, E_{BQ})$ defines an extension of $G-Q$. For simplicity, we also refer to this extension as $(Q, E_{BQ})$. Hence, we also use $\sourceOp$, $\sinkOp$, and $\connOp$, as well as $\compOp_{G-Q}$ for these extensions. 

The next reduction rule identifies a condition under which a component $Q$ can never be part of a solution, namely, if $Q$ includes strongly connected components with no in-neighbors or no out-neighbors, which is exactly the case if the empty set is in $\source{Q,E_{BQ}}$ or $\sink{Q,E_{BQ}}$.

\begin{reduction*}
\label{red:scc_no_sources_or_sinks}
  Let $Q$ be a component of $G - B$ such that $\induced{G}{Q}$ is not strongly connected. If $\emptyset \in \source{Q,E_{BQ}} \cup \sink{Q,E_{BQ}}$, include $Q$ into $O$. 
\end{reduction*}
\begin{proof}[Proof of Safeness]
    Since $\induced{G}{Q}$ is not strongly connected, $Q$ cannot be a solution by itself. Suppose $\emptyset \in \source{Q,E_{BQ}}$, the other case is analogous. Then, there is a source component $C$ in $\scc{\induced{G}{Q}}$ that has no incoming edge, neither in $\induced{G}{Q}$, nor in $E_{BQ}$. 
    Therefore, no vertex $v \in B$ can reach $C$.
    Because $Q$ can only be included as a whole and together with other vertices in $B$, $Q$ cannot be part of a solution.
\end{proof}

Note that after \Cref{red:in_out,red:scc_no_sources_or_sinks} have been applied exhaustively, every source in $Q$ has incoming edges from $B$, and every sink in $Q$ has outgoing edges to $B$. Finally, we have one more simple rule, which removes vertices $v \in O \setminus B$. It relies on the fact that by \Cref{red:in_out}, we also have $\nei{v} \subseteq O$.

\begin{reduction*}\label{red:remove_out}
  If \Cref{red:in_out} is not applicable, remove $O \setminus B$ from $G$.
\end{reduction*}

The previous reduction rules were useful to remove trivial cases and extend $I$ and $O$.
From now on, we assume that the instance is exhaustively reduced by \Cref{red:remove_out,red:scc_no_sources_or_sinks,red:in_out}.
Therefore, any component of $G - B$ is either contained in $I$ or does not intersect $I$ and $O$.

The next two rules will be twin type reduction rules that allow us to bound the number of remaining components.
If there are two components of $G-B$ that form equivalent extensions it is enough to keep one of them, since they fulfill the same role in forming a strongly connected subgraph. The reduction rules rely on \Cref{lem:source_sink_conn_equiv,lem:union_equiv_if_equiv} to show that equivalent extensions can replace each other and can be added to any solution.

\begin{reduction*}
\label{red:scc_twins}
  Let $Q_1, Q_2$ be components of $G - B$ such that both $\induced{G}{Q_1}$ and $\induced{G}{Q_2}$ are not strongly connected and $(Q_1, E_{BQ_1})$ and $(Q_2, E_{BQ_2})$ are equivalent. Delete $Q_2$ and increase the weight of some $q \in Q_1$ by $\w{Q_2}$. If $Q_2 \cap I \ne \emptyset$, set $I = I \cup Q_1$.
\end{reduction*}
\begin{proof}[Proof of Safeness]
  By the previous reduction rules, components of $G-B$ can only be included as a whole or not at all. Notice that since $\conn{Q_1, E_{BQ_1}} = \conn{Q_2, E_{BQ_2}}$ and \Cref{red:scc_no_sources_or_sinks}, we get $\nei{Q_1} = \nei{Q_2}$. Let $S$ be a solution to the old instance. We differentiate some cases.

  If $S \cap (Q_1 \cup Q_2) = \emptyset$, we know that also $\nei{S} \cap (Q_1 \cup Q_2) = \emptyset$. Therefore, the neighborhood size and strong connectivity of $S$ do not change in the new instance, and it is also a solution.
  If $S$ includes only one of $Q_1$ and $Q_2$, assume without loss of generality $S \cap (Q_1 \cup Q_2) = Q_1$, since $Q_1$ is not strongly connected by itself, the solution must include vertices of $B$. Because $\nei{Q_1} = \nei{Q_2}$, the solution must also include $Q_2$, a contradiction.
  If $S$ includes both $Q_1$ and $Q_2$, we claim that $S' \coloneqq S \setminus Q_2$ is a solution for the new instance. The neighborhood size and weight clearly remain unchanged. 
  Strong connectivity follows by \Cref{lem:union_equiv_if_equiv} and \Cref{lem:source_sink_conn_equiv}.
  The last two cases also show that if a solution had to include at least one of $Q_1$ and $Q_2$, that is, $(Q_1 \cup Q_2) \cap I \ne \emptyset$, any solution for the reduced instance must also include $Q_1$. Therefore, the adaptation to $I$ is correct.

  Let $S$ be a solution to the reduced instance. Again, if $S$ does not include vertices from $Q_1$, then $S$ will immediately be a solution to the old instance. If $Q_1 \subseteq S$, then $S' \coloneqq S \cup Q_2$ will be a solution to the old instance by \Cref{lem:union_equiv_if_equiv} and \Cref{lem:source_sink_conn_equiv}.
\end{proof}

For strongly connected components, the rule is different, since we have to acknowledge the fact that they can be a solution by themselves. For every such component we destroy strong connectivity of smaller-weight equivalent component, which can then be reduced by \Cref{red:scc_twins}.

\begin{reduction*}
\label{red:scc_twins_single}
  Let $Q_1, Q_2$ be components of $G - B$ that are also strongly connected with $(Q_1, E_{BQ_1})$ and $(Q_2, E_{BQ_2})$ equivalent and $\w{Q_1} \ge \w{Q_2}$.
  Add a vertex $q_2'$ with edges $\died{q_2}{q_2'}$ and $\died{q_2'}{v}$, for all $q_2 \in Q_2$ and $v \in \nei{Q_2}$. Set $\w{q_2'} = 0$.
\end{reduction*}
\begin{proof}[Proof of Safeness]
  Let $S$ be a solution of the old instance.
  If $S \cap (Q_1 \cup Q_2) = \emptyset$, then $S$ is also a solution for the new instance.
  If $S$ includes only one of $Q_1$ and $Q_2$, then we must have $S \in \set{Q_1, Q_2}$, so $S' \coloneqq Q_1$ is a solution for the new instance with $\w{S'} \ge \w{S}$.
  If $S$ includes both $Q_1$ and $Q_2$, adding $q_2'$ to $S$ obviously gives a solution of the same weight.

  For a solution of the reduced instance $S$, we can simply remove $q_2'$ if it is inside for a solution to the old instance.
\end{proof}

Using both \Cref{red:scc_twins,red:scc_twins_single} exhaustively makes sure that there are at most two components per extension equivalence class left.
The last rule compresses the remaining components to equivalent components of bounded size.

\begin{reduction*}
\label{red:scc_comp}
  Let $Q$ be a component of $G - B$ that is not equal to its compressed extension. Replace $Q$ by its equivalent compressed extension and set the weight such that $\w{Q'} = \w{Q}$. If $Q \cap I \ne \emptyset$, set $I = I \cup Q'$.
\end{reduction*} 
\begin{proof}[Proof of Safeness]
  Since $Q$ is not strongly connected, it can only be part of a solution that includes some vertices from $G-Q$. Thus, we can apply \Cref{lem:source_sink_conn_equiv,lem:comp_weakly_conn_equiv}, and the old instance and the new instance have exactly the same strongly connected subgraphs. Because of $\conn{Q, E_{BQ}} = \conn{Q', E_{BQ'}}$ and \Cref{red:scc_no_sources_or_sinks}, we get $\nei{Q} = \nei{Q'}$. Since $\w{Q'} = \w{Q}$, the rule is safe.
\end{proof}

This finally allows us to bound the size of $G-B$. In the next lemma, we summarize the progress of our reduction rules and apply the bounds from the previous section. Note that we need the stronger bound using the neighborhood of the components instead of simply $B$.

\begin{lemma}
\label{lem:scc_all_reductions}
  Let $\mathcal{Q}$ be a set of components of $G - B$ with total neighborhood size $h \coloneqq \abs{\bigcup_{Q \in \mathcal{Q}} \nei{Q}}$. Then we can reduce the instance, or there are at most $2^{2^{h+1} + h^2}\left(2^{h+1} + h^2\right) = 2^{\bigO{2^h}}$ vertices in $\mathcal{Q}$ in total.
  
  Executing the reduction rules in the proposed order guarantees termination after $\bigO{n}$ applications. The total execution takes at most $n^{\bigO{1}}$ time.
\end{lemma}
\begin{proof}
  After applying \Cref{red:scc_twins,red:scc_comp} exhaustively, there can be at most one copy of every compressed extension among the components in $\mathcal{Q}$.
  Since the compressed component only depends on the neighborhood of a component, we can treat $\bigcup_{Q \in \mathcal{Q}} \nei{Q}$ as the base graph for our extensions. By \Cref{lem:range_size_comp}, there are at most $\abs{\mathcal{Q}} \le 2^{2\cdot 2^h + h^2}$ components. Also, each component has at most $2^{h+1} + h^2$ vertices, which immediately gives the claimed bound on the total number of vertices in $\mathcal{Q}$.
  
  Now we consider the running time. Most reduction rules clearly make progress, either by increasing the size of $I$ or $O$ or by decreasing the number of vertices in $G$. While compression does not always decrease the number of vertices, it will only be applied once per component. Finally, \Cref{red:scc_twins_single} decreases the number of strongly connected components in $G-B$, which also will never be increased again by \Cref{lem:comp_weakly_conn_equiv}. All of these progress measures are bounded by $n$, proving the first claim.

  For the second part, we only have to bound the runtime of a single application of a reduction rule. Most rules are clearly executable in polynomial time. The total size of the source, sink, and connection sets are also bounded by $\abs{\e{G}}$ and thus computable in polynomial time. Therefore, the compressed component can also be constructed in polynomial time. Note that this increases the size of our graph, but if we apply \Cref{red:scc_comp} only after all other reduction rules are completed, the compressed component never has to be considered again.
\end{proof}

\begin{lemma}[\protect{\cite[Lemma~3]{golovach2020finding}}]\label{lem:separation_algo}
  Given an undirected graph $G$, there is an algorithm with runtime $2^{\bigO{\min\{q,k\}\log(q+k)}}n^{\bigO{1}}$ that either finds a $(q,k)$-separation of $G$ or correctly reports that $G$ is $((2q+1)q2^k, k)$-unbreakable.
\end{lemma}

Now, we have all that it takes to solve our intermediate problem. 
The main idea of the algorithm is to shrink $B$ to a bounded size, by solving the problem recursively. Once $B$ is bounded, we apply our reduction rules by viewing components of $G-B$ as extensions, removing redundant equivalent extensions and compressing them. Thereby, we also bound the size and number of components of $G-B$ in terms of $\abs{B}$ using \Cref{lem:scc_all_reductions}. By choosing suitable constants, we can show that this decreases the total size of $G$, which will make progress to finally reduce it to the unbreakable case.

\begin{algorithm}[t]
   \caption{The recursive understanding algorithm for \scsborder{}\label{alg:rec_und_scc}.}
  \DontPrintSemicolon
  \SetKwFunction{FMain}{Solve}
  \SetKwProg{Fn}{def}{:}{}
  \Fn{\FMain{$G$, $I$, $O$, $B$, $\wOp$, $k$, $T$}}{
    $q\gets 2^{2^{2^{ck^2}}}$ for a suitable constant $c$\;
    \eIf{$G$ is $((2q+1)q2^k,k)$-unbreakable}
    {
      \KwRet solve the problem using \Cref{thm:unbreakable_scc}\;
    }{
      $(U,W) \gets $ $(q,k)$-separation of $G$ with $\abs{T \cap W \setminus U} \le k$\;
      $(\tilde{G}, \tilde{I}, \tilde{O}, \tilde{B}, \tilde{\wOp}, k, \tilde{T}) \gets $ restriction to $W$ with $\tilde{T} = (T \cap W) \cup (U \cap W)$\;
      $\mathcal{R} \gets $ \FMain{$\tilde{G}$, $\tilde{I}$, $\tilde{O}$, $\tilde{B}$, $\tilde{\wOp}$, $k$, $\tilde{T}$}\;
      $\mathcal{N} \gets \tilde{T} \cup \bigcup_{R \in \mathcal{R}} \nei{R} \cap W$\;
      $\hat{B} \gets (B \cap U) \cup (B \cap \mathcal{N})$\;
      $(G^*, I^*, O^*, \hat{B}^*, \wOp^*, k^*, T^*) \gets $ reduce $(G,I,O,\hat{B},\wOp,k,T)$ with \Cref{lem:scc_all_reductions}\;
      \KwRet \FMain{$G^*$, $I^*$, $O^*$, $\hat{B}^*$, $\wOp^*$, $k^*$, $T^*$}\;
    }
  }
\end{algorithm}

\begin{theorem}
\label{thm:border_scc_fpt}
  \scsborder{} can be solved in time $2^{2^{2^{\bigO{k^2}}}}n^{\bigO{1}}$.
\end{theorem}
\begin{proof}

  Let $\mathcal{I} = (G, I, O, B, \wOp, k, T)$ be our \scsborder{} instance. See \Cref{alg:rec_und_scc} for a more compact description of the algorithm. A high level display of the approach can be found in \Cref{fig:recursive_calls}.
  Define $q = 2^{2^{2^{ck^2}}}$ for a constant $c$. We later show that a suitable $c$ must exist.
  \zzcommand{\neighbor}{\mathcal{N}}

  First, we run the algorithm from \Cref{lem:separation_algo} on the underlying undirected graph of $G$ with $q$ and $k$. If it is $((2q+1)q2^k, k)$-unbreakable, we solve the instance directly using the algorithm from \Cref{thm:unbreakable_scc}.

  Therefore, assume that we have a $(q,k)$-separation $(U,W)$. Without loss of generality, since $\abs{T} \le 2k$ we can assume that $\abs{T \cap W \setminus U} \le k$. Thus, we can construct a new instance to solve \emph{the easier side} of the separation. Take $\tilde{G} = \induced{G}{W}, \tilde{I} = I \cap W, \tilde{O} = O \cap W, \tilde{B} = B \cap W$, write $\tilde{\wOp}$ for the restriction of $\wOp$ to $W$, and set $\tilde{T} = (T \cap W) \cup (U \cap W)$. Since $\abs{U \cap W} \le k$, also $\abs{\tilde{T}} \le 2k$ holds and $\tilde{I} \coloneqq (\tilde{G}, \tilde{I}, \tilde{O}, \tilde{B}, \tilde{\wOp}, k, \tilde{T})$ is a valid instance, which we solve recursively.

  Let $\mathcal{R}$ be the set of solutions found in the recursive call. For $R \in \mathcal{R}$, define $N_R = \nei{R} \cap W$. Define $\mathcal{N} = \tilde{T} \cup \bigcup_{R \in \mathcal{R}} N_R$. 

  We now restrict $B$ in $W$ to use only vertices in the neighborhood that have been neighbors in a solution in $\mathcal{R}$, that is, only vertices in $\mathcal{N}$. Define $\hat{B} = (B \cap U) \cup (B \cap \mathcal{N})$. We now replace $B$ in our instance with $\hat{B}$ and apply all our reduction rules exhaustively to arrive at the instance $(G^*, I^*, O^*, \hat{B}^*, \wOp^*, k^*, T^*)$. Finally, we also solve this instance recursively and return the solutions after undoing the reduction rules.

 \subparagraph*{Correctness.}
  We already proved that the reduction rules and the algorithm for the unbreakable case are correct. The main statement we have to show is that we can replace $B$ with $\hat{B}$ without throwing away important solutions. That means that the instances $(G,I,O,B,\wOp,k,T)$ and $\hat{\mathcal{I}} \coloneqq (G,I,O,\hat{B},\wOp,k,T)$ are equivalent in the sense that any solution set for one instance can be transformed to a solution set to the other instance with at least the same weights. This justifies solving $\hat{\mathcal{I}}$ instead of $\mathcal{I}$.
  Consider the boundary complementation  $\mathcal{I}' \coloneqq (G',I',O',B',\wOp',k')$ and $\hat{\mathcal{I}}' \coloneqq (G',I',O',\hat{B}',\wOp',k')$ that are caused by the same $(X,Y,Z,R)$.
  To show the claim, we consider the two directions. Any solution for $\hat{\mathcal{I}}'$ is immediately a solution for the same boundary complementation for $\mathcal{I}'$ since the only difference is that we limit the possible neighborhood to a subset of $B$.

  For the other direction, consider a solution $S$ to $\mathcal{I}'$, and we want to show that there is a solution $\hat{S}$ to $\hat{\mathcal{I}}'$ using only vertices of $\hat{B}'$ in the neighborhood with $\w{\hat{S}} \ge \w{S}$. If $S \cap W = \emptyset$, then $S$ is also immediately a solution for $\hat{\mathcal{I}}'$. Therefore, assume $S \cap W \ne \emptyset$. Define $\tilde{X} \coloneqq \tilde{T} \cap S$, $\tilde{Y} \coloneqq \tilde{T} \cap \neiG{G - (W\setminus U)}{S}$, and $\tilde{Z} \coloneqq \tilde{T} \setminus (\tilde{X} \cup \tilde{Y})$. Let $\tilde{R}$ be the set of $(a,b) \in \tilde{X} \times \tilde{X}$ such that there is an $a$-$b$-path in $\induced{G}{S \setminus {W \setminus U}}$. Finally, let $\tilde{k} \coloneqq \abs{\nei{S} \cap W}$. Thus, we can construct a boundary complementation instance $(\tilde{G}', \tilde{I}', \tilde{O}', \tilde{B}', \tilde{w}', \tilde{k})$ from $(\tilde{X}, \tilde{Y}, \tilde{Z}, \tilde{R})$. One can easily verify that $(S \cap \ve{\tilde{G}'}) \cup \set{u_r}{r \in \tilde{R}}$ is a feasible solution to this instance. Furthermore, the maximum solution $\tilde{S} \in \mathcal{R}$ to this instance gives a new set $\hat{S} \coloneqq (\tilde{S} \cap W) \cup (S \cap U) \subseteq \ve{G'}$ that has at least the same weight as $S$. We also know that $\nei{\hat{S}} \subseteq \hat{B}'$ and $\abs{\nei{\hat{S}}} \le k$. See \Cref{fig:scc_bc} for a visualization of the construction corresponding to a solution $S$.

  All that remains is to verify that $\hat{S}$ is strongly connected. For $v_1, v_2 \in \tilde{S} \cap W$, we can simply use the $v_1$-$v_2$-path in $\tilde{S}$, replacing subpaths via $u_r$ for $r \in \tilde{R}$ with the actual paths in $S \cap U$. If $v_1 \in \tilde{S} \cap W$ and $v_2 \in S \cap U$, we can walk to any $x \in \tilde{X}$ which must have a path to $v_2$ in $S$, in which may need to replace subpaths via $W$. This can be done, since $\tilde{S} \cap W$ connects all pairs of $x_1, x_2 \in X$ that are not connected via $S \cap U$. The two remaining cases follow by symmetry, justifying the replacement of $B$ with $\hat{B}$.

  Finally, we have to show that the recursion terminates for the right choice of $c$; that is, both recursively solved instances have strictly smaller sizes than the original graph.
  For the first recursive call, note that the boundary complementation adds at most $k^2 \le q$ vertices to $\tilde{G}$. Since $\abs{U \setminus W} > q$, we have $\abs{\ve{\tilde{G}}} < \abs{\ve{G}}$.

  For the second recursive call, since $\abs{\tilde{T}} \le 2k$ and by \Cref{lem:number_border_complementations}, we have $\abs{\mathcal{N}} \le 2k + (k+1)k2^{c_1k^2} \le 2^{c_2k^2}$ for constants $c_1$ and $c_2$.
  After applying all our reduction rules, by \Cref{lem:scc_all_reductions} for a suitable choice of $c_3$ and $c$ we get \[\abs{W^*} \le \abs{\neighbor} + 2^{2^{\abs{\neighbor}+1} + \abs{\neighbor}^2}\left(2^{\abs{\neighbor}+1} + \abs{\neighbor}^2\right) \le 2^{c_3 2^{\abs{\neighbor}}} \le 2^{c_32^{2^{c_2k^2}}} \le 2^{2^{2^{ck^2}}} \eqqcolon q.\]
  Since before the reductions, we had $\abs{W \setminus U} > q$, the reduced graph $G^*$ also has fewer vertices than $G$. Therefore, this recursive call also makes progress and the recursion terminates.

  \subparagraph*{Runtime.} 
  We follow the analysis of~\cite{chitnis2016designing}. With \Cref{lem:separation_algo}, we can test
  if the undirected version of $G$ is unbreakable in time $2^{k\log{q+k}}n^{\bigO{1}} \le 2^{2^{2^{\bigO{k^2}}}} n^{\bigO{1}}$.
  If the graph turns out to be unbreakable, the algorithm of \Cref{thm:unbreakable_scc} solves it in time $2^{\bigO{k^2\log{q}}} \le 2^{2^{2^{\bigO{k^2}}}}n^{\bigO{1}}$.
  Executing the reduction rules takes polynomial time in $n$ by \Cref{lem:scc_all_reductions}. 

  All that is left to do is analyze the recursion. Let $n' \coloneqq \abs{W}$. By the separation property, we know that $q < n' < n-q$. From the correctness section, we get $\abs{\ve{G^*}} \le \abs{\ve{G}} - n' + q$. 
  Note that the recursion stops when the original graph is $(q,k)$-unbreakable, which must be the case if $\abs{\ve{G}} \le 2q$.
  We arrive at the recurrence 
  \[
    T(n) = \begin{cases}2^{2^{2^{\bigO{k^2}}}}, & \text{for $n \le 2q$;}\\
    \left(\max_{q < n' < n-q} T(n' + k^2) + T(n - n' + q)\right) + 2^{2^{2^{\bigO{k^2}}}}n^{\bigO{1}},  & \text{otherwise.} \end{cases}
  \]
  Notice that $2^{2^{2^{\bigO{k^2}}}}$ appears in every summand of the expanded recurrence and can be ignored here and multiplied later. Furthermore, $n^{\bigO{1}}$ is bounded from above by a convex polynomial. Therefore, it is enough to consider the extremes of the maximum expression. For $n' = q+1$, the first recursive call evaluates to $T(q+1+k^2) \le T(2q) \le 2^{2^{2^{\bigO{k^2}}}}$. The second call only introduces an additional factor of $n$. For $n' = n-q-1$, the second expression evaluates to $T(2q+1)$ which is clearly also bounded by $2^{2^{2^{\bigO{k^2}}}}$. Thus, we arrive at the final runtime of $2^{2^{2^{\bigO{k^2}}}}n^{\bigO{1}}$.
\end{proof}

Finally, we can use \scsborder{} to solve \scs{}. Since in our original problem every vertex could be part of the solution or its neighborhood, we initially set $I \coloneqq O \coloneqq \emptyset$ and $B \coloneqq \ve{G}$. Furthermore, we only want to consider the boundary complementation that changes nothing, which we achieve by setting $T \coloneqq \emptyset$.

\test*

\section{Secluded \texorpdfstring{$\mathcal{F}$}{F}-Free subgraph and Secluded DAG}
\label{chap:hardness}

\zzcommand{\prob}{\textsc{Out-Secluded $\mathcal{F}$-Free WCS}}
In this section, we show that \textsc{Out-Secluded $\mathcal{F}$-Free Subgraph} is W[1]-hard for almost all choices of $\mathcal{F}$, even if we enforce weakly connected solutions. Except for a few missing cases, we establish a complete dichotomy when \textsc{Out-Secluded $\mathcal{F}$-Free Weakly Connected Subgraph} (\prob{}) is hard and when it is FPT. This is a surprising result compared to undirected graphs and shows that out-neighborhood behaves completely different. 
The same problem was studied before in undirected graphs, where it admits FPT-algorithms using recursive understanding~\cite{golovach2020finding} and branching on important separators~\cite{jansen2023single}. For total neighborhood, both algorithms still work efficiently since the seclusion condition acts exactly the same as in undirected graphs. However, the result changes when we look at out-neighborhood, where we show hardness in \Cref{thm:f-free-hard-always} for almost all choices of $\mathcal{F}$. 
To formalize for which kind of $\mathcal{F}$ the problem becomes W[1]-hard, we need one more definition.
\begin{definition}[Inward Star]
  We say that a directed graph $G$ is an \emph{inward star}, if there is one vertex $v \in \ve{G}$ such that $\e{G} = \set{\died{u}{v}}{u \in \ve{G} \setminus \set{v}}$, that is, the underlying undirected graph of $G$ is a star and all edges are directed towards the center.
\end{definition}

Due to the nature of the construction, we show that if no inward star contains any $F \in \mathcal{F}$ as a subgraph, the problem becomes W[1]-hard. Besides this restriction, the statement holds for all non-empty $\mathcal{F}$, even families with only a single forbidden induced subgraph. Later, we explain how some other cases of $\mathcal{F}$ can be categorized as FPT or W[1]-hard.



\restateffree*
\begin{proof}
\begin{figure}
  \centering
  \hfill
  \begin{subfigure}{0.3\textwidth}
    \centering
    \includegraphics[width=0.7\textwidth,page=2]{figures/incidence_graph_reduction}
    \caption{A \textsc{Clique} input graph $G$.}
  \end{subfigure}
  \hfill
  \begin{subfigure}{0.54\textwidth}
    \centering
    \includegraphics[width=0.9\textwidth,page=1]{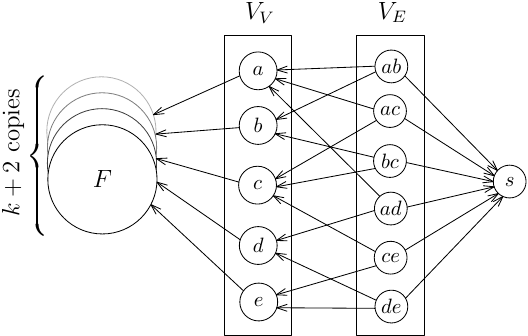}
    \caption{The output graph $G'$ of our reduction function for $G$.}
  \end{subfigure}
  \hfill
  \caption{The construction of the reduction in the proof of \Cref{thm:f-free-hard-always}. Shaded vertices represent a complete graph on $k+2$ copies of $F$, and every $v \in V_V$ is connected to every copy. No solution in $G'$ can include a vertex in a copy of $F$ or $V_V$. Thus, any secluded subset $S$ of $V_E$ of size $\binom{k}{2}$ corresponds to the clique $\outNei{S}$ in $G$.}\label{fig:clique_reduction}
\end{figure}

    Let $F \in \mathcal{F}$. We give a reduction from \textsc{Clique}, inspired by~\cite{fomin2013parameterized}. Given an undirected graph $G$, $k$, and $w$, consider the incidence graph, which we modify in the following. We add $k+2$ copies of $F$ and a single vertex $s$. For every edge $e = \uned{u}{v} \in E$, we add the edges $\died{e}{u}$, $\died{e}{v}$, and $\died{e}{s}$. Furthermore, we add an edge from every vertex $v \in V_V$ to every vertex in every copy of $F$. For two different copies $F_1$ and $F_2$ of $F$, we add edges in both directions between every vertex $v_{f1} \in F_1$ and $v_{f2}$ in $F_2$. Denote this newly constructed graph by $G'$. Finally, set $k' \coloneqq k$ and $w \coloneqq \binom{k}{2} + 1$.
    See \Cref{fig:clique_reduction} for a visualization of the graph construction.
    Notice that $k \le \abs{\ve{G}}$, therefore the construction has polynomial size, and $k' + w$ depends only polynomially on $k$.

    We show that there is a clique of size $k$ in $G$ if and only $G'$ has a weakly connected $k$-out-secluded $\mathcal{F}$-free subgraph of weight at least $w$.
    If $C \subseteq \ve{G}$ forms a $k$-clique in $G$, consider the set of vertices $S$ in $G'$ consisting of $s$ and all edge vertices $v_{\uned{a}{b}}$ where both endpoints $a$ and $b$ lie in the clique $C$. This is clearly weakly connected and has size $\binom{k}{2}+1$. Also, we have $\outNei{S} = C$ and $\abs{C} = k$.

    For the other direction, consider a solution $S$ for $G'$. Notice that no vertex from the copies of $F$ or from $V_V$ can be in $S$; otherwise, $S$ has at least $k+1$ out-neighbors. To ensure the connectivity, if $k > 1$, $s$ must be in $S$. Hence, the out-neighborhood of $S$ must be part of $V_V$. Because $\abs{\outNei{S}} \le k$ and $\abs{S} \setminus \set{s} = \binom{k}{2}$, clearly $S$ must consist of the edges of a clique in $G$.
\end{proof}

Note that \Cref{thm:f-free-hard-always} is very general and excludes many properties $\propOp$ immediately from admitting FPT-algorithms. 
Another interesting example is finding secluded DAGs, a very natural extension of secluded trees, studied in~\cite{golovach2020finding,donkers2023finding}. There we choose $\mathcal{F}$ to be the set of all directed cycles. Although this set is not finite, hardness follows from \Cref{thm:f-free-hard-always}.

\restatedag*

Finally, we analyze some of the more remaining cases in the following paragraphs.

\subparagraph*{Trivial Cases}
If $\mathcal{F}$ contains the graph with only a single vertex, the empty set is the only $\mathcal{F}$-free solution. If $\mathcal{F}$ contains two vertices connected by a single edge, weakly connected solutions can only consist of a single vertex or contain bidirectional edges $(u,v)$ and $(v,u)$ for $u,v \in \ve{G}$. Both of these problems are clearly solvable in FPT-time, the second one via our algorithm for \textsc{Secluded Clique}. Additionally, if $\mathcal{F} = \emptyset$, we can clearly choose the maximum weight component of $G$ as the solution.

\subparagraph*{Independent Sets} 
An independent set is a subgraph of an inward star. One kind of graph that cannot be included in $\mathcal{F}$ such that \Cref{thm:f-free-hard-always} shows hardness are independent sets. Surprisingly, the problem becomes FPT if $\mathcal{F}$ includes an independent set of any size. First, notice that independent sets that are part of $\mathcal{F}$ can be ignored if there is a smaller independent set in $\mathcal{F}$. Suppose the size of the smallest independent set is $\alpha+1$. This means that any solution must be an $\alpha$-bounded graph, i.e., a graph without an independent set of size $\alpha$.

We have shown already how these can be enumerated efficiently with the algorithm in the proof of \Cref{thm:alpha_bounded_fpt}. For every enumerated subgraph, we can simply check if it is $\mathcal{F}$-free in time $\abs{\mathcal{F}}n^{\|\mathcal{F}\|}$. Therefore, this case is also FPT with parameter $k$ if $\mathcal{F}$ is finite. 
\begin{theorem}
  Let $\mathcal{F}$ be a finite family of directed graphs that contains an independent set. Then \prob{} is solvable in time $\abs{\mathcal{F}}(2\alpha + 2)^kn^{\|\mathcal{F}\| + \alpha + \bigO{1}}$.
\end{theorem}

\subparagraph*{Inward Stars}
Consider what happens when $\mathcal{F}$ contains an inward star $F$ with two leaves. Then, the weakly connected $F$-free graphs are exactly the rooted trees where the root could be a cycle instead of a single vertex, with potentially added bidirectional edges. For $F \in \mathcal{F}$, we do not have a solution, but we conjecture that the FPT branching algorithm for \textsc{Secluded Tree} by \cite{donkers2023finding} can be transferred to the directed setting. 

If $\mathcal{F} = \set{F}$ for a star $F$ with more than two leaves, the problem remains W[1]-hard. We can slightly modify the construction in the proof of \Cref{thm:f-free-hard-always} such that there is not one extra vertex $s$, but one extra vertex $s_e$ for every $e \in \e{G}$. We connect all $s_e$ internally in a directed path. A solution for \prob{} can then include the whole extra path and the edges encoding the clique. This avoids inward star structures with more than two leaves, showing hardness in this case. For only one single forbidden induced subgraph $F$, we therefore conjecture, that \prob{} is FPT with parameter $k$ if and only if $F$ has no edges or is an inward star with at most two leaves. For all other cases, we have proven W[1]-hardness with parameter $k+w$ for unit weights.

\subparagraph*{Remaining Cases}
The previous paragraphs resolve almost all possible cases for $\mathcal{F}$, except for a few cases which are difficult to characterize. For example, we could have $\{F_s, F_p\} \subseteq \mathcal{F} $, where $F_s$ is an inward star and $F_p$ is a path, which makes a new connecting construction instead of $s$ and the $s_e$ necessary. The above characterization is still enough to give an understanding of which cases are hard and which are FPT for all natural choices of $\mathcal{F}$.

\section{Secluded Subgraphs with Small Independence Number}
\label{chap:tour}
\zzcommand{\clique}{\textsc{Secluded Clique}}
\newcommand{\aboundedlong}{\textsc{Out-Secluded $\alpha$-Bounded Subgraph}}
\newcommand{\abounded}{\textsc{Out-Secluded $\alpha$-BS}}
\newcommand{\secuna}[1]{\textsc{Undirected Secluded {#1}-BS}}

In this section, we consider a generalization of the undirected \clique{} problem to directed graphs. We generalize this property in the following way.

\begin{definition}[$\alpha$-Bounded]
    A directed graph $G$ is called \emph{$\alpha$-bounded} if $G$ includes no independent set of size $\alpha + 1$, that is, for all $S \subseteq \ve{G}$ with $\abs{S} = \alpha + 1$, the graph $\induced{G}{S}$ includes at least one edge.
\end{definition}

Our problem of interest will be the \aboundedlong{} problem (\abounded{}).
Note that $\alpha$ is part of the problem and therefore a constant.
In the directed case, $\alpha = 1$ is analogous to the \prname{Out}{Tournament} problem, except that tournaments cannot include more than one edge between a pair of vertices.

In this section, we first show how to solve \abounded{} with a branching algorithm in \Cref{sec:alpha_bounded}.
 The general nature of these branching rules allows us to transfer and optimize them. In \Cref{sec:clique}, we significantly improve the best known runtime for the \clique{} problem to $1.6181^kn^{\bigO{1}}$.

\subsection{Secluded \texorpdfstring{$\alpha$}{α}-Bounded Subgraphs}\label{sec:alpha_bounded}

Our branching algorithm starts by selecting one part of the solution and then relies on the fact that the remaining solution has to be close around the selected part. More concretely, we start by picking an independent set that is part of the solution and build the remaining solution within its two-hop neighborhood. The following lemma justifies this strategy. 
The proof of this lemma was sketched on \emph{Mathematics Stack Exchange}~\cite{stackexchange}, we give an inspired proof again for completeness.

\begin{lemma}
\label{lem:ind_set_reaches}
  In every directed graph $G$, there is an independent set $U \subseteq \ve{G}$, such that every vertex in $\ve{G} \setminus U$ is reachable from an $u \in U$ via a path of length at most $2$.
\end{lemma}
\begin{proof}
  We prove the statement by induction on $\abs{\ve{G}}$. For $\abs{\ve{G}}=1$, it clearly holds.
  Assume the statement holds for all graphs with fewer than $\abs{\ve{G}}$ vertices, we want to prove it for $G$.
  Let $v \in \ve{G}$ be a vertex. If $v$ reaches every other vertex via at most 2 edges, $\set{v}$ is our desired independent set. Otherwise, $T \coloneqq \ve{G} \setminus \coutNei{v}$ is non-empty. Since $\abs{T} < \abs{\ve{G}}$, we can apply the induction hypothesis. So, let $T' \subseteq T$ be an independent set in $\induced{G}{T}$ that reaches every vertex of $T$ via at most two edges. We consider the edges between $v$ and $T'$. 

  Since $T \cap \outNei{v} = \emptyset$, there cannot be an edge $(v,t')$ for any $t' \in T'$. If there is an edge $(t', v)$ for some $t' \in T'$, then $T'$ is also a solution for $G$ since $t'$ can reach $\coutNei{v}$ via at most 2 edges.
  Finally, if there is no edge between $v$ and $T'$, then $T' \cup \set{v}$ is an independent set. Also $T' \cup \set{v}$ reaches by definition all of $T \cup \coutNei{v}$ via at most 2 edges.
\end{proof}

For $\alpha$-bounded graphs, all independent set have size at most $\alpha$, so trying every subset of $\ve{G}$ of size at most $\alpha$ allows us to find a set that plays the role of $U$ in \Cref{lem:ind_set_reaches} in our solution. Hence, the first step of our algorithm will be iterate over all such sets. From then on, we only consider solutions $S \subseteq \coutNei{\coutNei{U}}$.

Then, our algorithm uses two branching rules. Both rules have in common that in contrast to many well-known branching algorithms~\cite[Chapter~3]{cygan2015parameterized}, we never include vertices into a partial solution. Instead, due to our parameter, we only decide that vertices should be part of the final neighborhood. This means that we remove the vertex from the graph and decrease the parameter by 1.
Our first branching rule branches on independent sets of size $\alpha + 1$ in $\coutNei{\coutNei{U}}$, to ensure that $\coutNei{\coutNei{U}}$ becomes $\alpha$-bounded. The second rule is then used to decrease the size of the neighborhood of $\coutNei{\coutNei{U}}$, until it becomes secluded.

To formalize which vertices exactly to branch on, we use some new notation. For a vertex $v \in \ve{G}$ and a vertex set $U \subseteq \ve{G}$, pick an arbitrary shortest path from $U$ to $v$ if there exists one. Define $\shor{v}$ to be the vertices on that path including $v$ but excluding the first vertex $u \in U$. Note that after removing vertices from the graph, $\shor{v}$ might change.

Now, we can give our algorithm for finding secluded $\alpha$-bounded graphs. 

\newcommand{\curr}{\coutNei{\coutNei{U}}}
\restateab*
\begin{proof}
    \begin{figure}[t]
      \centering \hfill
      \begin{subfigure}{0.48\textwidth} \centering
        \includegraphics[height=0.5\textwidth,page=8]{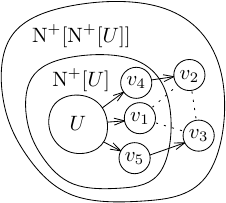}
        \caption{Case 1. If $\alpha \le 2$, not all of $v_1$, $v_2$, and $v_3$ can be in $S$. One of the five vertices must be in $\nei{S}$.}
      \end{subfigure} \hfill
      \begin{subfigure}{0.48\textwidth} \centering
        \includegraphics[height=0.5\textwidth,page=9]{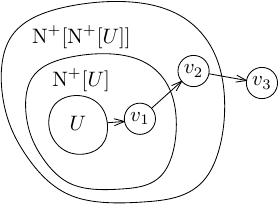}
        \caption{Case 2. Since $v_3$ is not reachable from $U$ via 2 edges, one of $v_1$, $v_2$, and $v_3$ must be in $\nei{S}$.}
      \end{subfigure} \hfill
      \caption{A visualization of the branching rules for the \abounded{} algorithm in \Cref{thm:alpha_bounded_fpt}. We are only looking for solutions $S$ including $U$, where every $v \in S$ is reachable from some $u \in U$ via at most 2 edges. Dotted connections stand for edges that do not exist.}\label{fig:alpha}
    \end{figure}

    \begin{algorithm}[!ht]
      \caption{The branching algorithm for \abounded{} that returns a solution including the set $U \subseteq \ve{G}$.}
      \label{alg:abounded}
      \DontPrintSemicolon
      \SetKwFunction{FMain}{$\alpha$-BS}
      \SetKwProg{Fn}{def}{:}{}
      \Fn{\FMain{$G$, $\omega$, $w$, $k$, $U$}}{
        \uIf{$k < 0$} {
          \textbf{abort}\; 
        }
        \uElseIf{$\curr$ is a solution} {
            \KwRet $\curr$\;
        }
        \uElseIf{there is an independent set $I \subseteq \curr$ of size $\abs{I} = \alpha + 1$} {
            \ForEach{$v \in \bigcup_{w \in I} \shor{w}$} {
                Call \FMain{$G-v$, $\omega$, $w$, $k-1$, $U$}\;
            }
        }
        \ElseIf{there is $w \in \outNei{\curr{}}$} {
            \ForEach{$v \in \shor{w}$} {
                Call \FMain{$G-v$, $\omega$, $w$, $k-1$, $U$}\;
            }
        }
      }
    \end{algorithm}

    Let $(G,\wOp, w, k)$ be an instance of \abounded{}. 
    We guess a vertex set $U \subseteq \ve{G}$ that should be part of a desired solution $S$. Furthermore, we want to find a solution $S$ to the instance such that $S \subseteq \coutNei{\coutNei{U}}$, that is, every $v \in S$ should be reachable from $U$ via at most two edges.
    We give a recursive branching algorithm that finds the optimal solution under these additional constraints.
    The algorithm is also described in \Cref{alg:abounded}.
    
    When deciding that a vertex should be part of the final neighborhood, we can simply delete it and decrease $k$ by one. In case $k$ decreases below 0, there is no solution. If $\curr$ is a solution to the instance, we return it. Otherwise we apply the following branching rules and repeat the algorithm for all non-empty independent sets $U \subseteq \ve{G}$ of size at most $\alpha$.
    
    \begin{description}
        \item[Case 1. $\curr$ is not $\alpha$-bounded.] In this case, there must be an independent set $I \subseteq \curr$ of size $\alpha + 1$. Clearly, not all of $I$ can be part of the solution, so there is a vertex $w \in I \setminus S$. This means that either $w$ or a vertex on every path from $U$ to $w$ must be in $\outNei{S}$.
        The set $\shor{w}$ is one such path of length at most 2. Thus, one of $\bigcup_{w \in I} \shor{w}$ must be part of the out-neighborhood of $S$ and we branch on all of these vertices. For one vertex, delete it and decrease $k$ by 1.
    
        \medskip
        \item[Case 2. $\curr$ is $\alpha$-bounded, but has additional neighbors.] Consider one of these neighbors $w \in \outNei{\curr}$.
        Since $w$ is not reachable from $U$ via at most two edges, we should not include it in the solution. Now, $\shor{w}$ is a path of length 3, and one of its vertices must be in $\outNei{S}$. We again branch on all options, delete the corresponding vertex and decrease $k$ by 1.
    \end{description}

    By \Cref{lem:ind_set_reaches}, there is a suitable choice of $U$ for every solution. Therefore, if we can find the maximum solution containing $U$ if one exists in every iteration with our branching algorithm, the total algorithm is correct.
    Also notice that the branching rules are a complete case distinction; if none of the rules apply, the algorithm reaches a base case.
    The remaining proof of correctness follows from a simple induction.

    We initially have to consider all subsets of $\ve{G}$ of size at most $\alpha$ while rejecting subsets that are not an independent set in time $n^{\alpha+1}$.
    To bound the runtime of the branching algorithm, notice that in each case we branch and make progress decreasing $k$ by 1.
    In the first cases, the independent set $I$ has size $\alpha + 1$ for every $w \in I$, the path $\shor{w}$ includes at most 2 vertices. Hence, there are at most $2(\alpha + 1)$ branches in this case. The second case gives exactly 3 branches and is thus dominated by the first rule.  This gives the claimed runtime and concludes the proof. 
\end{proof}

For the tournament setting, we can simply extend the first branching rule from the algorithm from \Cref{thm:alpha_bounded_fpt} to also branch on two vertices that are connected by a bidirectional edge.

\begin{theorem}
  \prname{Out}{Tournament} is solvable in time $4^kn^{\bigO{1}}$.
\end{theorem}

Moreover, note that if we instead define the directed problem to ask for the total neighborhood $\nei{S} \le k$, the problem is identical to solving the undirected version on the underlying undirected graph. 
In undirected graphs, instead of \Cref{lem:ind_set_reaches}, any maximal independent set reaches every vertex via a single edge. We can use this insight to adapt the previous algorithm to a faster one for the total neighborhood variant of the problem.

\restateabtotal*
\begin{proof}[Proof Sketch]
    Notice that in every directed graph $G$, any maximal independent set $U \subseteq \ve{G}$ fulfills $\cnei{U} = V$. Therefore, we can execute the same algorithm as described in \Cref{alg:abounded} and the proof of \Cref{thm:alpha_bounded_fpt} with $\cnei{U}$ instead of $\coutNei{\coutNei{U}}$. This leads to a smaller size set of vertices to branch on, namely $\alpha + 1$ instead of $2\alpha + 2$ for the first case and 2 instead of 3 for the second case. The remaining steps work out in the same way.
\end{proof}

\subsection{Faster \textsc{Secluded Clique} using Branching}\label{sec:clique}

The \clique{} problem has been shown to admit an FPT-algorithm running in time $2^{\bigO{k \log k}}n^{\bigO{1}}$ by contracting twins and then using color coding~\cite{golovach2020finding}. Furthermore, the single-exponential algorithm for \textsc{Secluded $\mathcal{F}$-Free Subgraph} from~\cite{jansen2023single} can also solve this problem in time $2^{\bigO{k}}n^{\bigO{1}}$ when choosing $\mathcal{F}$ to include only the independent set on two vertices. 
In this section, we give a faster and simpler algorithm using branching that achieves the same in time $1.6181^k n^{\bigO{1}}$.


The ideas behind our branching rules are the the same as in \Cref{thm:alpha_bounded_fpt}. However, this time, the forbidden structures are only independent sets of size 2. Therefore, it is enough guess a single vertex $u$ initially. This allows us also to analyze the rules more closely and split them into several cases.
By differentiating between more different scenarios in these two high-level rules, we arrive at a branching vector of $(1,2)$ for the desired runtime.

    \begin{algorithm}[t]
      \caption{The $1.6181^kn^{\bigO{1}}$ branching algorithm for \clique{} that returns a $k$-secluded clique of weight at least $w$ including $u \in \ve{G}$.}
      \label{alg:clique_fast}
      \DontPrintSemicolon
      \SetKwFunction{FMain}{Clique}
      \SetKwProg{Fn}{def}{:}{}
      \Fn{\FMain{$G$, $\omega$, $w$, $k$, $u$}}{
        \uIf{$k < 0$} {
          \textbf{abort}\; 
        }
        \uElseIf{$\cnei{u}$ is a solution} {
            \KwRet $\cnei{u}$\;
        }
        \uElseIf{there are distinct $v_1, v_2, v_3 \in \nei{u}$ with $\{v_1, v_2\}, \{v_1, v_3\} \notin \e{G}$} {
            Call \FMain{$G-v_1$, $\omega$, $w$, $k-1$, $u$}\;
            Call \FMain{$G-\{v_2, v_3\}$, $\omega$, $w$, $k-2$, $u$}\;
        }
        \uElseIf{there are distinct $v_1, v_2 \in \nei{u}$ with $\{v_1, v_2\} \notin \e{G}$, $\nei{v_1}, \nei{v_2} \subseteq \cnei{u}$, and $\w{v_1} \le \w{v_2}$} {
            Call \FMain{$G-v_1$, $\omega$, $w$, $k-1$, $u$}\;
        }
        \uElseIf{there are distinct $v_1, v_2 \in \nei{u}$ with $\{v_1, v_2\} \notin \e{G}$ and $\nei{v_1} \not\subseteq \cnei{u}$} {
            Call \FMain{$G-v_1$, $\omega$, $w$, $k-1$, $u$}\;
            Call \FMain{$G - (\neiOp'(v_1) \cup \set{v_2})$, $\wOp$, $w$, $k- (\abs{\neiOp'(v_2)} + 1)$, $u$}\;
        }
        \uElseIf{there is $v_1 \in \nei{u}$ with $\neiOp'(v_1) = \set{v_2}$} {
            Call \FMain{$G - v_2$, $\wOp$, $w$, $k-1$, $u$}\;
        }
        \ElseIf{there is $v \in \nei{u}$ with $\abs{\neiOp'(v)} \ge 2$} {
            Call \FMain{$G-v$, $\wOp$, $w$, $k-1$, $u$}\;
            Call \FMain{$G - \neiOp'(v)$, $\wOp$, $w$, $k- \abs{\neiOp'(v)}$, $u$}\;
        }
      }
    \end{algorithm}
\restateclique*
\begin{proof}

    \begin{figure}[t] \centering \hfill
      \begin{subfigure}{0.32\textwidth} \centering
        \includegraphics[width=0.6\textwidth,page=3]{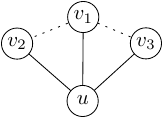}
        \caption{Case 1a. Either $v_2 \in \nei{S}$ or $v_1,v_3 \in \nei{S}$ must hold.}
      \end{subfigure} \hfill
      \begin{subfigure}{0.32\textwidth} \centering
        \includegraphics[width=0.6\textwidth,page=4]{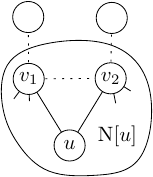}
        \caption{Case 1b. It only changes the weight if $v_1 \in \nei{S}$ or $v_2 \nei{S}$.}
      \end{subfigure} \hfill
      \begin{subfigure}{0.32\textwidth} \centering
        \includegraphics[width=0.6\textwidth,page=5]{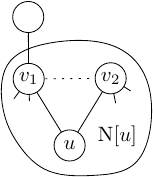}
        \caption{Case 1c. Either $v_1 \in \nei{S}$ or $\neiOp'(v_1) \cup \set{v_2} \subseteq \nei{S}$ must hold.}
      \end{subfigure} \hfill
      \begin{subfigure}{0.49\textwidth} \centering
        \includegraphics[width=0.38\textwidth,page=6]{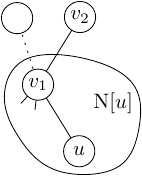}
        \caption{Case 2a. The only neighbor of $v_1$ outside $\cnei{u}$ is $v_2$. Therefore, $v_2 \in \nei{S}$ must hold.}
      \end{subfigure} \hfill
      \begin{subfigure}{0.49\textwidth} \centering
        \includegraphics[width=0.38\textwidth,page=7]{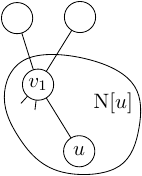}
        \caption{Case 2b. Either $v_1 \in \nei{S}$ or $\neiOp'(v_1) \subseteq \nei{S}$ must hold.}
      \end{subfigure} \hfill
      \caption{A visualization of the branching rules for the improved \textsc{Secluded Clique} algorithm in \Cref{thm:clique_better}. We are only looking for solutions $S$ including $u$. Dotted connections stand for edges that do not exist. Also, vertices outside of $\cnei{u}$ are never connected to $u$. In (b) to (e), we can assume that $v_1$ is connected to all other vertices in $\cnei{u}$, and the same holds for $v_2$ in (b) and (c). Dotted edges to the outside of $\cnei{u}$ mean that no such edges exist.}\label{fig:clique2}
    \end{figure}
    
    Let $(G,\wOp, w, k)$ be an instance of \clique{}. 
    We guess a vertex $u\in \ve{G}$ that should be part of a desired solution $S$, that is we look for a solution $S \subseteq \cnei{u}$ with $u \in S$. We give a recursive branching algorithm that branches on which vertices should be part of the neighborhood. 
    The algorithm is also described in \Cref{alg:clique_fast}.
    
    When deciding that a vertex should be part of the final neighborhood, we can simply delete it and decrease $k$ by one. In case $k$ decreases below 0, there is no solution. If $\cnei{u}$ is a solution to the instance, we return it. These are the base cases of our algorithm. Otherwise we apply the following branching rules and repeat the algorithm for all choices of $u \in \ve{G}$.
    Denote for $v \in \ve{G}$ with $\neiOp'(v) \coloneqq \nei{v} \setminus \cnei{u}$ the additional neighborhood of $v$ that is not part of $\cnei{u}$. 
    The branching rules are visualized in \Cref{fig:clique2}.
    
\begin{description}
    \item[Case 1. $\cnei{u}$ does not form a clique.] In this case we must have a pair of vertices  $v_1, v_2 \in N(X)$ with $\set{v_1, v_2}\notin E(G)$. We now break into following three subcases.
    
    \begin{itemize}
        \item\label{it:clique1} \textbf{1a. There is a vertex $v_1 \in \nei{u}$ not connected to at least 2 others.} Formally speaking, there are distinct $v_2, v_3 \in \nei{u}$ with $\set{v_1, v_2}, \set{v_1, v_3} \notin \e{G}$. In this case, we know that if $v_1$ is in $S$, then both $v_2$ and $v_3$ cannot be in $S$.
        Therefore, we branch into removing both $v_2$ and $v_3$ or removing just $v_1$. We decrease $k$ by $1$ or $2$.

        \item \label{it:clique2}\textbf{1b. There are disconnected $v_1, v_2 \in \nei{u}$ that have no other neighbors.}
        Since the previous subcase does not apply and $\neiOp'(v_1) = \neiOp'(v_2) = \emptyset$, we know that $v_1$ and $v_2$ must be twins. Not both of them can be part of $S$, so we could remove either of them. Therefore, it is optimal to remove the vertex with smaller weight, decrease $k$, and recurse.
        
        \item \label{it:clique3}\textbf{1c. There are disconnected $v_1, v_2 \in \nei{u}$ with at least one other neighbor.}
        Suppose $\neiOp'(v_1) \ne \emptyset$. In this case, we know that if $v_1 \in S$, then both $v_2$ and all of $\neiOp'(v_1)$ have to go in $\nei{S}$.
        Thus, we branch into two cases where we either remove $v_1$ or $\neiOp'(v_1) \cup \set{v_2}$ and decrease $k$ accordingly.
    \end{itemize} \medskip 
    
    \item[Case 2. $\cnei{u}$ forms a clique.] In this case, we break into following two subcases.
    \begin{itemize}
        \item \label{it:clique4}\textbf{2a. There is $v_1 \in \nei{u}$ with only one other neighbor $\neiOp'(v_1) = \set{v_2}$.} In this case, it is always optimal to include $v_1$ in $S$. This adds at most one neighbor $v_2$ compared to not including $v_1$, which would add $v_1$ into $\nei{S}$. Since weights are non-negative, $v_2$ should go into $\nei{S}$. We remove it from $G$ and decrease $k$ by 1.
        
        \item \label{it:clique5}\textbf{2b. There is $v \in \nei{u}$ with at least two other neighbors in $\neiOp'(v)$.} Secondly, if there is $v$ remaining with at least two neighbors, we again consider all options.
        If $v \in S$, we know that all of $\neiOp'(v)$ have to be in $\nei{S}$. We branch into two cases in which we either remove $v$, or remove all of $\neiOp'(v)$ from $G$. In both cases, we decrease $k$ by the number of removed vertices.
    \end{itemize}
\end{description}

Notice that the case distinction is exhaustive, that is, if none of the cases apply $\cnei{u}$ is a clique without neighbors.
To give an upper bound on the runtime of the algorithm, notice that in each of the cases we either reduce the size of the parameter by at least one (Case 1b, Case 2a) or do branching. But in each of the branching steps (Case 1a, Case 1c, Case 2b) we make progress decreasing $k$ by 1 in the first branch, and at least 2 in the second branch. Hence they have the branching vector $(1,2)$, which is known~\cite[Chapter 3]{cygan2015parameterized} to result in a runtime of $1.6181^k$.  Therefore, the runtime of the recursive algorithm is upper-bounded by $1.6181^k n^{\mathcal{O}(1)}$, where the final runtime is achieved by multiplying by $n$ (the number of initial guesses of the vertex $u$). The proof of correctness follows from a simple induction. This concludes the algorithm. 
\end{proof}

\bibliography{main}

\end{document}